\documentclass[12pt,reqno]{amsart}

\usepackage{amsthm,amsmath,amsfonts,amssymb,amsxtra,appendix,bookmark,dsfont,bm,color, braket}
\usepackage{cite}
\usepackage{mathtools}

\newtheorem{theorem}{Theorem}[section]
\newtheorem{lemma}[theorem]{Lemma}

\theoremstyle{definition}
\newtheorem{definition}[theorem]{Definition}
\newtheorem{remark}[theorem]{Remark}

\newtheorem{proposition}[theorem]{Proposition}

\renewcommand{\epsilon}{\varepsilon}
\newcommand{\E}{\mathcal{E}}

\renewcommand{\phi}{\varphi}
\newcommand{\R}{\mathbb{R}}
\renewcommand{\S}{\mathbb{S}}

\newcommand{\be}{\begin{equation}}

\newcommand{\1}{{\ensuremath {\mathds 1} }}

\newcommand\tr{\mathop{\mathrm{tr}}\nolimits} 

\numberwithin{equation}{section}  

\begin{document}

\title{The maximal excess charge in reduced Hartree-Fock molecule}

\author{Yukimi Goto}
\address{RIKEN iTHEMS, Wako, Saitama 351-0198, Japan}
\email{\tt yukimi.goto@riken.jp}

\begin{abstract}
We consider a molecule described by the Hartree-Fock model without the exchange term.
We prove that nuclei of total charge $Z$ can bind at most $Z+C$ electrons, where $C$ is a constant independent of $Z$.
\end{abstract}

\maketitle

\section{Introduction}
We denote by $N>0$ and $K>0$ the total number of electrons and nuclei, respectively.
Our model is described by an energy functional defined on one-body density matrices.
A one-body density matrix $\gamma$ is a self-adjoint operator on $L^2(\R^3)$ satisfying $0 \le \gamma \le 1$ and $\tr \gamma < \infty$.
The kernel can be written as $\gamma(x, y) = \sum_{i\ge1} n_i \phi_i(x) \phi^*_i(y)$, with the eigenfunctions $\phi_i$, such that $\gamma \phi_i = n_i \phi_i$.
Then we define the one-particle electron density $\rho_\gamma$ by $\rho_{\gamma} (x)= \gamma(x, x)$.
The reduced Hartree-Fock (RHF) functional is given by the functional
\[
\mathcal{E}^\mathrm{RHF}(\gamma)
=\tr \left[\left( -\frac{1}{2} \Delta - V_Z \right) \gamma\right]
+D[\rho_{\gamma}],
\]
where
\[
D[\rho_{\gamma}]  \coloneqq \frac{1}{2}\iint_{\R^3\times \R^3}
\frac{\rho_{\gamma}(x) \rho_{\gamma}(y)}{|x-y|}dxdy.
\]
Here $V_Z$ is the Coulomb potential
\[ V_Z(x) = \sum_{i=1}^K \frac{z_i}{|x-R_i|}, \quad Z=\sum_{i=1}^K z_i,
\]
where $z_1, \dots, z_K >0$ are the charges of fixed  nuclei located at $R_1, \dots, R_K \in \R^{3}$.
For all $N> 0$ and $z_i > 0$, we define the energy by
\[
E^\mathrm{RHF}(N, Z) = \inf \{\mathcal{E}^\mathrm{RHF}(\gamma) \colon
\gamma \in \mathcal{P},
\tr \gamma = N\}
\]
where $\mathcal{P} = \{ \gamma \colon \gamma =\gamma^{\dagger}, 0 \le \gamma \le 1, (-\Delta +1)^{1/2}\gamma (-\Delta +1)^{1/2} \in \mathcal{S}^1\}$, and $\mathcal{S}^1$ is the set of trace-class operators. 

Our interest is to investigate the maximum ionization.
It is believed (see~\cite[Chapter 12]{LS}) that real atoms in nature can only bind one or possibly two extra electrons.
This ionization conjecture has only been shown for the atomic case ($K=1$) in the reduced Hartree-Fock model~\cite{SolovejRHF} and the full Hartree-Fock model~\cite{SolovejIC}.
Recently, Frank et al. proved this conjecture also in the Thomas-Fermi-Dirac-von Weizs\"aker model~\cite{TFDWIC} and the M\"uller model~\cite{mullerIC}.
However, they only dealt with the atomic case.

In this article, we will prove as follows.
\begin{theorem}[Maximal ionization]
\label{theorem.main}
Let $z_\mathrm{min} \coloneqq \min_{1\le j \le K}z_j$, $z_\mathrm{max} \coloneqq \max_{1\le j \le K} z_j$ and $R_\mathrm{min} = \min_{i \neq j}|R_i - R_j|$.
We assume $z_\mathrm{min}  \ge \delta z_\mathrm{max}$, and $R_\mathrm{min} \ge c_0$ with some $c_0, \delta>0$ independent of $Z$.
There is a constant $C_K>0$ depending on $K$ such that for all $Z>0$, if $E^\mathrm{RHF}(N, Z) $ has a minimizer, then $N \le Z+C_K$ holds.
\end{theorem}

\begin{remark}
Presumably, the true $C_K$ behaves linearly on $K$, but this is still open.
\end{remark}

As in the atomic cases in~\cite{SolovejRHF, SolovejIC, TFDWIC, mullerIC}, the basic strategy to prove Theorem~\ref{theorem.main} is Solovej's argument, which consists of three main ingredients:
\begin{enumerate}
\item An inequality to control the electronic density outside a ball by using the screened potential inside the ball.

\item A Sommerfeld estimate for the screened nuclear potential in the Thomas-Fermi (TF) theory.

\item A bootstrap argument to compare the Hartree-Fock potential to the Thomas-Fermi potential.
\end{enumerate}

Recently, Samojlow has generalized these ingredients to diatomic ($K=2$) molecules~\cite{Samojlow}, where the Born-Oppenheimer curves were investigated.
From a technical point of view, our paper's main novelty is to extend Samojlow's results to $K>2$.
However, Samojlow has restricted the analysis to the neutral case $z_1=z_2=N/2$, and thus the arguments in~\cite{Samojlow} do not rely on the exterior $L^1$-estimate in a region $A_r \coloneqq \{x \in \R^3 \colon |x-R_j| > r \, \text{ for all } j=1, \dots, K\}$ with an adequate $r \in (0, R_\mathrm{min})$.
Indeed, in the neutral case $N=Z$, we can write
\[
\int_{A_r}\rho^\mathrm{RHF} =  \int_{A_r}\rho^\mathrm{TF}
+ \int_{A_r^c}\left(\rho^\mathrm{TF} - \rho^\mathrm{RHF} \right),
\]
where $\rho^\mathrm{RHF}$ and $\rho^\mathrm{TF}$ are the ground-state densities in the RHF and TF models, respectively.
The first term can be estimated by the Sommerfeld bound.
For the second term, we may use the bound (\ref{ite1}) below.
Hence, for the first ingredient, we need a new strategy to control the number of electrons in $A_r$.
One of our analyses' central ideas is to combine Lieb's method~\cite{Lieb1984} and the moving plane method~\cite{TFDWIC, mullerIC}.
Namely, on the first step we will use Lieb's method to control $\int_{A_{R_0}}\rho^\mathrm{RHF}$, where $R_0 \coloneqq \min \{1, R_\mathrm{min}/4\}$.
Next, the moving plane method allows us to control the $L^1$-norm of the density in the regions $r \le |x-R_j| \le R_0$ for all $j=1, \dots, K$.

For the second ingredient: Usually, the Sommerfeld asymptotics refers to the behavior at infinity, but it has been proven in~\cite[Theorem 4.6]{SolovejIC} for sufficiently large $|x|$.
Thus, to extend the bounds to molecular cases, we have to study bounds on small domains close to nuclei.
Then one can extend the proof in~\cite{Samojlow} to the case of $K>2$ with some slight modifications.
The only essential difference is that instead of the using the half spaces $H^{\pm} \coloneqq \{ x \in \R^3 \colon \pm (x -R\nu/2)\cdot R\nu/2 > 0\}$ with $R_1=0$ and $R_2=R\nu$ for some $\nu \in \S^2$, we will use the Voronoi cells $\Gamma_j \coloneqq \{ x \in \R^3 \colon |x-R_j| < |x-R_i| \, \text{ for all } i\neq j\}$ for $j=1, \dots, K$.

For the third ingredient: Our strategy is virtually the same as in~\cite{TFDWIC, SolovejIC, mullerIC, Samojlow}.
At the technical level, the molecular cases are slightly more complicated than the atomic or diatomic ones. Thus, we may require additional arguments.
In particular, the (sub) harmonicity of our potentials will be crucial.

\section*{Outline}
This article is organized as follows.
In Section 2, we derive the exterior estimate for the number of electrons in $A_r$.
In Section 3, we compare our minimizer with the minimizer of an effective exterior functional.
In Section 4, we study TF theory for molecules. In particular, we prove Sommerfeld bounds.
The proof of Theorem~\ref{theorem.main} is given in Section 7 by using Solovej's argument relying on an initial step given in Section 5 and an iteration step in Section 6.

\section*{Conventions}
In the remainder of this article, we will assume that $E^\mathrm{RHF}(N, Z) $ has a minimizer $\gamma^\mathrm{RHF}$ for some $N \ge Z$.
Then we will write $\rho^\mathrm{RHF} \coloneqq \rho_{\gamma^\mathrm{RHF}}$ and $\rho^\mathrm{TF}$ to be the minimizer for the neutral TF molecule.
We also use the shorthand notation
\[
D(f, g) \coloneqq \frac{1}{2}\iint_{\R^3 \times \R^3} \frac{f(x)g(y)}{|x-y|} \, dxdy.
\]


\section{$L^1$ exterior estimate}
As in~\cite{{RuskaiSolovej}}, we choose smooth localizing functions $\theta_j \in C^\infty(\R^3)$, $j = 0, 1, \dots, K$ with the following properties:
\begin{definition}
\label{localize}
Let $\lambda \in (0, 1/2]$ and $R_0 \coloneqq \min \{1, R_\mathrm{min}/4\}$.

\begin{enumerate}
  \setlength{\parskip}{0.1cm}
    \setlength{\itemsep}{0.1cm}
\item[(i)]
For $j \ge 1$ we have $\theta_j(x) = \theta (|x-R_j|/R_0)$, with smooth $\theta $ satisfying $0 \le \theta \le 1$ and $\theta(t) = 1$ if $t < 1$ and $\theta(t) = 0$ if $t > 1+\lambda$.

\item[(ii)]
$\sum_{j=0}^K \theta_j(x)^2 = 1$ (which defines $\theta_0$).

 These properties imply
\item[(iii)]
$|\nabla \theta_j(x)| \le CR_0^{-1}$ for all $j$.
\end{enumerate}

We put  $\gamma_j \coloneqq \theta_j \gamma^\mathrm{RHF} \theta_j$ and $\rho_j\coloneqq \rho_{\gamma_j}$ for $j = 0, 1 \dots, K$.

\end{definition}
Now we introduce here the screened potentials defined by
\begin{align*}
\Phi^\mathrm{RHF}_{r}(x) &\coloneqq V_Z(x)-  \int_{A_r^c} \frac{\rho^\mathrm{RHF}(y)}{|x-y|} \, dy,\\
\Phi^\mathrm{TF}_{r}(x) &\coloneqq V_Z(x)-  \int_{A_r^c} \frac{\rho^\mathrm{TF}(y)}{|x-y|} \, dy,
\end{align*}
where $A_r^c$ stands for the complement of $A_r = \{x \in \R^3 \colon |x-R_j| > r \, \text{ for all } j=1, \dots, K\}$.
Our first goal is to control the integral $\int_{A_{R_0}} \rho^\mathrm{RHF}$.
Namely, we will show as follows.
\begin{lemma}
\label{theorem:outside}
Let
\[
\phi(x) \coloneqq \sum_{j=1}^K\mu_j|x-R_j|^{-1}, \quad \mu_j = \frac{z_j}{Z}.
\]
Then it holds that
 \begin{align*}
\left(  \int_{A_{R_0}} \rho^\mathrm{RHF}(x) dx \right)^2 
  &\le  C\left( \frac{1}{R_0^2} + \sup_{x \in  A_{R_0/3}}  \phi(x)^{-1} \left[ \Phi^\mathrm{RHF}_{R_0/3}(x) \right]_+\right) \int_{A_{R_0/3}}\rho^\mathrm{RHF}.
 \end{align*}
\end{lemma}

\begin{proof}
The reduced Hartree-Fock minimizer $\gamma^\mathrm{RHF} = \sum_{i=1}^\infty \lambda_i \ket{u_i} \bra{u_i}$ satisfies the RHF equation $H_{\gamma^\mathrm{RHF}} u_i = \epsilon_i u_i$ with $\epsilon_i \le 0$ (see~\cite[Theorem 1]{SolovejRHF}).
Here $H_{\gamma^\mathrm{RHF}}$ is defined by
\begin{align*}
H_{\gamma^\mathrm{RHF}} = -\frac{1}{2} \Delta -V_Z(x) + \rho^\mathrm{RHF} \star |x|^{-1}.
\end{align*}
Now we use Lieb's method in~\cite{Lieb1984}.
By the RHF equation, we have
\begin{align*}
0 &\ge \sum_{i=1}^\infty  \epsilon_i \lambda_i \int |u_i(x)|^2 \phi(x)^{-1} \theta_0(x)^2\, dx \\
&= \sum_{i=1}^\infty  \frac{\lambda_i}{2} \int \nabla (u_i(x)^\ast \phi(x)^{-1} \theta_0(x)^2) \cdot \nabla u_i(x) dx
- \int \rho_0V_Z \phi^{-1} \\
&\quad+ \iint \frac{\rho^\mathrm{RHF}(x) \rho^\mathrm{RHF}(y)}{|x-y|} \phi(x)^{-1} \theta_0(x)^2 \, dx dy.
\end{align*}

Next, we use the next proposition.
\begin{proposition}[The IMS formula~{\cite[Lemma 2.4]{SolovejIC}}]
For $u \in H^1(\R^3)$ and $\eta \in C^1(\R^3)$ satisfying $\|\nabla \eta \|_\infty \le C$ we have
\[
\mathrm{Re} \int \nabla (\eta^2 u^\ast) \cdot \nabla u = \int |\nabla u|^2 - \int |\nabla \eta|^2|u|^2.
\]
\end{proposition}
Then we deduce that
\[
\int \nabla (u_i(x)^\ast \phi(x)^{-1} \theta_0(x)^2) \cdot \nabla u_i(x) dx
=  \int |\nabla (u_i(x)\phi(x)^{-1/2} \theta_0(x)|^2 dx - \int |\nabla(\theta_0 \phi^{-1/2})|^2|u_i|^2.
\]
By definition, $|\nabla \theta_0 \phi^{-1/2}|^2 \le CR_0^{-2}$ holds.
Hence
\[
\int \nabla (u_i(x)^\ast \phi(x)^{-1} \theta_0(x)^2) \cdot \nabla u_i(x) dx
\ge -\frac{C}{R_0^2} \int_{A_{R_0}} |u_i(x)|^2 \, dx.
\]
We note from the triangle inequality that
\[
\phi(x)^{-1} + \phi(y)^{-1} = \sum_{j=1}^K \mu_j\frac{|x-R_j| + |y-R_j|}{\phi(x)\phi(y)|x-R_j||y-R_j|}
\ge \sum_{j=1}^K \frac{\mu_j|x-y|}{\phi(x)\phi(y)|x-R_j||y-R_j|}.
\]
Then it holds that
\begin{align*}
\iint \frac{\rho^\mathrm{RHF}  (x) \rho^\mathrm{RHF}  (y)  }{|x-y|}& \phi(x)^{-1} \theta_0(x)^2 \, dx dy \\
&= \iint \frac{\rho^\mathrm{RHF}  (x)\rho^\mathrm{RHF}  (y)  }{|x-y|} \phi(x)^{-1} (1-\theta_0(y)^2)\theta_0(x)^2 \, dx dy \\
&\quad + \frac{1}{2} \iint \frac{\rho^\mathrm{RHF} (x) \rho^\mathrm{RHF} (y) }{|x-y|} (\phi(x)^{-1} + \phi(y)^{-1}) \theta_0(y)^2\theta_0(x)^2 \, dx dy \\
&\ge  \iint \frac{\rho^\mathrm{RHF} (x)\rho^\mathrm{RHF} (y)  }{|x-y|} \phi(x)^{-1} (1-\theta_0(y)^2)\theta_0(x)^2 \, dx dy \\
&\quad + \frac{1}{2}\sum_{j=1}^K\mu_j\left(  \iint \frac{\rho_0(x) dx }{\phi(x)|x-R_j|}\right)^2.
\end{align*}
Furthermore, we may estimate
\begin{align*}
 \iint &\frac{\rho^\mathrm{RHF} (x)\rho^\mathrm{RHF} (y) }{|x-y|} \phi(x)^{-1} (1-\theta_0(y)^2)\theta_0(x)^2 \, dx dy \\
 &\ge \sum_{j=1}^K  \iint_{|y-R_j| < R_0} \frac{\rho^\mathrm{RHF} (x)\rho^\mathrm{RHF} (y) }{|x-y|} \phi(x)^{-1} \theta_0(x)^2 \, dx dy.
\end{align*}
These estimates lead to that
\begin{align*}
0 &\ge -\frac{C}{R_0^2} \int_{A_{R_0}}\rho^\mathrm{RHF} \, dx - C\int \rho_0(x)\phi(x)^{-1} \Phi^\mathrm{RHF}_{R_0/2}(x)\\
&\quad +\frac{1}{2}\sum_{j=1}^K\mu_j\left(  \int \frac{\rho_0(x) dx }{\phi(x)|x-R_j|}\right)^2.
\end{align*}
Furthermore, by the convexity, we deduce from $\sum_{j=1}^K\mu_j(\phi(x)|x-R_j|)^{-1} = 1$ that
\[
\sum_{j=1}^K \mu_j \left(  \int \frac{\rho_0(x) dx }{\phi(x)|x-R_j|}\right)^2
 \ge \left(  \int \rho_0(x) dx \right)^2.
\] 
Together with these estimates, we have
\begin{align*}
 \left(  \int_{A_{R_0}} \rho_0(x) dx \right)^2 &\le \frac{C}{R_0^2} \int_{A_{R_0}}\rho^\mathrm{RHF}(x) \, dx  \\
 &\quad +C\int_{A_{R_0}} \rho^\mathrm{RHF}(x)\phi(x)^{-1} [\Phi^\mathrm{RHF}_{R_0}(x)]_+\, dx.
 \end{align*}
 Hence we arrive at
 \begin{align*}
  \frac{1}{2} \left(  \int_{A_{(1+\lambda)R_0}} \rho^\mathrm{RHF}(x) dx \right)^2 
  &\le  \frac{C}{R_0^2} \int_{A_{R_0}}\rho^\mathrm{RHF} \, dx \\
&\quad+ C \sup_{x \in A_{R_0}} \phi(x)^{-1} \left[ \Phi^\mathrm{RHF}_{R_0}(x) \right]_+ \int_{A_{R_0}}\rho^\mathrm{RHF}.
 \end{align*}
 Replacing $R_0$ to $(1+\lambda)^{-1}R_0$ and choosing $\lambda = 1/2$, we have the claim.
\end{proof}
Following, we will use the cut-off function
\[
\chi_{r}^+ = \1_{A_r}
\]
and a smooth function $\eta_{r} \colon \R^3 \to [0, 1]$ satisfying
\[
\chi_{r}^+ \ge \eta_{r} \ge \chi_{(1+\lambda)r}^+, \quad |\nabla \eta_{r}| \le C(\lambda r)^{-1}.
\]
The next lemma is a modification of~\cite[Lemma 7]{mullerIC} and ~\cite[Lemma 5]{Kehle2017}.
\begin{lemma}
\label{outside}
For all $r \in (0, R_0]$, $s > 0$, and for all $\lambda \in (0, 1/2]$ we have
\begin{align*}
 \int_{A_r} \rho^\mathrm{RHF}(x) dx 
&\le
C\sum_{j=1}^K \int_{r \le|x-R_j|<(1+\lambda)r} \rho^\mathrm{RHF} (x) \, dx\\
&\quad+  C\left(\sup_{x \in A_r} \phi(x)^{-1}[\Phi_{r}^\mathrm{RHF}(x)]_+  + s + (\lambda^2 s)^{-1}  +\lambda^{-1} +  \frac{1}{R_0^2}\right)  \\
&\quad+C\left(s^2 \tr (- \Delta \eta_{r} \gamma^\mathrm{RHF} \eta_{r}) \right)^{3/5}.
\end{align*}
\end{lemma}
\begin{proof}
As \cite[Corollary 1]{Kehle2017}, we can obtain the binding inequality
\[
E^\mathrm{RHF}(N, Z) \le E^\mathrm{RHF}(N-M, Z) + E^\mathrm{RHF}(M, 0) \quad \text{for any } M>0.
\]
For fixed $\lambda \in (0, 1/2]$, and any $s, l > 0$, $\nu \in \S^2$ we choose
\[
\chi^{(i)}_j (x) = g_i \left(\frac{\nu \cdot h_j(x) - l}{s}\right), \quad i=1, 2.
\]
Here $g_i \colon \R \to \R$ satisfy
\[
g_1^2 + g_2^2 = 1, \quad g_1(t) = 1 \, \text{ if } \, t \le 0, \quad \mathrm{supp} \, g_1 \subset \{t \le 1\}, \quad |\nabla g_1| + |\nabla g_2| \le C,\]
and $h_j \colon \R^3 \to \R^3$ is the function with $|h_j(x)| \le |x-R_j|,$ $h_j(x) = 0$  if  $|x-R_j| \le r$; $h_j(x) = x-R_j$  if  $|x-R_j| \ge (1+\lambda)r,$ and $|\nabla h_j(x)| \le C\lambda^{-1},$  $j=1, \dots, K$.
We denote $\gamma_j^{i} \coloneqq \chi^{i}_j \gamma_j \chi^{i}_j$ for $j=1, \dots, K$ and $i=1, 2$, where $\gamma_j$ is as in Definition~\ref{localize}.
We note that the supports of $\gamma_j$, $j=1, \dots, K,$  are mutually disjoint by definitions.
Then, by using the IMS formula, we have
\begin{align*}
\E^\mathrm{RHF}(\gamma) &\le \E^\mathrm{RHF}\left(\sum_{j=1}^K\gamma_j^{(1)} \right) + \E^\mathrm{RHF}_{V_Z=0}( \gamma_0)+ \sum_{j=1}^K\E^\mathrm{RHF}_{V_Z=0}(\gamma_j^{(2)}) \\
&= \sum_{j=1}^K \sum_{i=1,2}\E^\mathrm{RHF}(\gamma_j^{(i)}) + \E^\mathrm{RHF}(\gamma_0)  +\sum_{1\le i<j\le K}2D(\rho_{\gamma_i^{(1)}}, \rho_{\gamma_j^{(1)}})  \\
&\quad+ \sum_{j=1}^K\tr(V_Z\gamma_j^{(2)}) + \tr(V_Z\gamma_0)\\
&=
\sum_{j=0}^K \E^\mathrm{RHF}(\gamma_j) +\sum_{1\le i<j\le K}2D(\rho_{\gamma_i^{(1)}}, \rho_{\gamma_j^{(1)}}) + \sum_{j=1}^K\tr(V_Z\gamma_j^{(2)}) 
+\tr(V_Z\gamma_0)\\
&\quad + \sum_{j=1}^K\left( \sum_{i=1,2}\int |\nabla \chi_j^{(i)}|^2\rho_j-\iint \frac{\chi_j^{(2)}(x)^2\rho_j(x)\rho_j(y)\chi_j^{(1)}(y)^2}{|x-y|} \, dxdy \right).
\end{align*}
Again by the IMS formula, we arrive at
\begin{align*}
0 &\le \sum_{1\le i<j\le K}2D(\rho_{\gamma_i^{(1)}}, \rho_{\gamma_j^{(1)}}) + \sum_{j=1}^K\tr(V_Z\gamma_j^{(2)}) +  \tr(V_Z\gamma_0) \\
&\quad + \sum_{j=1}^K\left( \sum_{i=1,2}\int |\nabla \chi_j^{(i)}|^2\rho_j-\iint \frac{\chi_j^{(2)}(x)^2\rho_j(x)\rho_j(y)\chi_j^{(1)}(y)^2}{|x-y|} \, dxdy \right) \\
&\quad  +\sum_{j=0}^K\int |\nabla \theta_j|^2\rho^\mathrm{RHF} -\sum_{j=1}^K 2D(\rho_0, \rho_j)  -\sum_{1\le i < j \le K}2D(\rho_i, \rho_j).
\end{align*}
By constructions, we obtain
\[
2D(\rho_{\gamma_i^{(1)}}, \rho_{\gamma_j^{(1)}})  -2D(\rho_i, \rho_j) \le -2D(\rho_{\gamma_i^{(1)}}, \rho_{\gamma_j^{(2)}}) -2D(\rho_{\gamma_i^{(2)}}, \rho_{\gamma_j^{(1)}}),
\]
and
\[
\sum_{i=1, 2}\int |\nabla \chi_j^{(i)}|^2 \rho_j \le C(1 + (\lambda s)^{-2}) \int_{\nu \cdot h_j(x) - s \le l \le \nu \cdot h_j (x)} \rho_j(x) dx.
\]
We note that 
\begin{align*}
\tr(V_Z\gamma_0)  -\sum_{j=1}^K 2D(\rho_0, \rho_j)
&\le \int_{\R^3} \rho_0(x) \Phi_{R_0}^\mathrm{RHF}(x) \, dx.
\end{align*}
Then it follows that for all $j$
\begin{align*}
 &\int V_Z(x)\chi_j^{(2)}(x)^2 \rho_j(x)\, dx
- \sum_{i=1}^K\iint \frac{\chi_j^{(2)}(x)^2 \rho_j(x) \rho_i(y)\chi_i^{(1)}(y)^2}{|x-y|} \, dxdy\\
&\le \int \chi_j^{(2)}(x)^2 \rho_j(x) \Phi^\mathrm{RHF}_{r}(x) dx - \iint_{|y-R_j| \ge r} \frac{\chi_j^{(2)}(x)^2 \rho_j(x) \rho_j(y)\chi_j^{(1)}(y)^2}{|x-y|} \, dxdy\\
&\le
\int_{l \le \nu \cdot h_j(x)}\rho_j(x) [ \Phi^\mathrm{RHF}_{r}(x) ]_+ dx 
- \iint_{\substack{\nu \cdot h_j(y) \le l \le \nu \cdot h_j(x) -s}} \chi_r^+(y) \frac{\rho_j(x) \rho_j(y)}{|x-y|} \, dxdy. \\
\end{align*}
Since $h_j(x) = x-R_j$ when $|x-R_j| > (1+\lambda)r$, we get
\begin{align*}
&\iint_{\substack{\nu \cdot h_j(y) \le l \le \nu \cdot h_j(x) -s}} \chi_{r}^+(y) \frac{\rho_j(x) \rho_j(y)}{|x-y|} \, dxdy\\
&\quad \ge
\iint_{\substack{\nu \cdot (y-R_j) \le l \le \nu \cdot (x-R_j) -s}} \chi_{(1+\lambda)r}^+(y) \chi_{(1+\lambda)r}^+(x) \frac{\rho_j(x) \rho_j(y)}{|x-y|} \, dxdy.
\end{align*}
With  these inequality, we have that
\begin{equation}
\label{fundaineq}
\begin{split}
\sum_{j=1}^K&\iint_{\substack{\nu \cdot (y-R_j) \le l \le \nu \cdot (x-R_j) -s}} \chi_{(1+\lambda)r}^+(y)\chi_{(1+\lambda)r}^+(x) \frac{\rho_j(x) \rho_j(y)}{|x-y|} \, dxdy\\
&\le
C\sum_{j=1}^K\bigg[ (1 + (\lambda s)^{-2}) \int_{\nu \cdot h_j(x) - s \le l \le \nu \cdot h_j (x)} \rho_j(x) dx
+ \int_{l \le \nu \cdot h_j(x)}\rho_j(x) [\Phi_{r}^\mathrm{RHF}(x)]_+ dx\bigg]  \\
&\qquad \qquad +\frac{C}{R_0^2} \int_{A_{R_0}}\rho^\mathrm{RHF}
+ \int_{\R^3} \rho_0(x) \Phi_{R_0}^\mathrm{RHF}(x) \, dx.
\end{split}
\end{equation}
for all $s, l > 0$ and $\nu \in \S^2$.
Now we integrate (\ref{fundaineq}) over $R_0 > l > 0$, then average over $\nu \in \S^2$ and use
\[
\int_{\S^2} [\nu \cdot x]_+ \frac{d \nu}{4 \pi} = \frac{|x|}{4}, \quad \text{for all } x \in \R^3.
\]
For the left side, we also use Fubini's theorem and
\[
\int_0^\infty \left( \1(b \le l \le a-s) + \1(-a \le l \le -b-s) \right)\, dl
\ge \left[ [a-b]_+ - 2s\right]_+
\]
with $a= \nu \cdot (x-R_j)$, $b= \nu \cdot (y-R_j)$.
For the right side, we use the fact that $\{x \colon \nu \cdot h_j(x) \ge l\} \subset \{x\colon |x-R_j| \ge r\}$ by construction.
We note that $|x-R_j| \le \phi(x)^{-1}$ on $r \le |x-R_j| \le  (1+\lambda)R_0$ and $R_0 \le \phi(x)^{-1}$ in $A_{R_0}$.
Together with these facts, we find that
\begin{equation*}
\begin{split}
\frac{1}{8} \sum_{j=1}^K&\left( \int_{(1+\lambda)r \le  |x-R_j|\le R_0}  \rho^\mathrm{RHF} \right)^2 \\
&\le C\bigg( \sup_{x \in A_r}\phi(x)^{-1}[ \Phi^\mathrm{RHF}_{r}(x) ]_+ + s + (\lambda^2 s)^{-1} +\frac{1}{R_0^2} \bigg) \int_{A_r}\rho^\mathrm{RHF}(x) dx \\
&\qquad \qquad \quad +  CsD\left[\chi^+_r \rho^\mathrm{RHF}\right].
\end{split}
\end{equation*}
For the left side, we use
\[
\left( \int_{(1+\lambda)r \le  |x-R_j|\le R_0}  \rho^\mathrm{RHF} \right)^2
\ge \frac{1}{2}\left(\int_{r \le  |x-R_j|\le R_0}  \rho^\mathrm{RHF} \right)^2 -\left( \int_{r \le|x-R_j|\le (1+ \lambda)r}  \rho^\mathrm{RHF}\right)^2.
\]
For the right side, by the Hardy-Littlewood-Sobolev inequality and the Lieb-Thirring inequality,
\begin{equation*}
\begin{split}
D[\chi^+_r \rho^\mathrm{RHF}]
&\le C\|\chi^+_r \rho^\mathrm{RHF} \|^2_{L^{6/5}} \\
&\le C\|\chi^+_r \rho^\mathrm{RHF}\|^{7/6}_{L^1} \|\chi^+_r \rho^\mathrm{RHF} \|^{5/6}_{L^{5/3}} \\
&\le C\|\chi^+_r \rho^\mathrm{RHF} \|^{7/6}_{L^1}\left( \tr (- \Delta \eta_{r} \gamma^\mathrm{RHF}\eta_{r}) \right)^{1/2}.
\end{split}
\end{equation*}
Hence, by Lemma~\ref{theorem:outside}, we have
\begin{equation*}
\begin{split}
&\left(\sum_{j=1}^K \int_{r \le |x-R_j| \le R_0}  \rho^\mathrm{RHF}\right)^2 +\left(  \int_{A_{ R_0}} \rho^\mathrm{RHF}(x) dx \right)^2 \\
&\le
C\sum_{j=1}^K\left( \int_{r <|x-R_j|<  (1+ \lambda)r} \rho^\mathrm{RHF} \right)^2\\
&\quad+  C\left(\sup_{x \in A_r} \phi(x)^{-1}[ \Phi^\mathrm{RHF}_{r}(x) ]_+  + s + (\lambda^2 s)^{-1} + \lambda^{-1} +\frac{1}{R_0^2}\right) \int_{A_r}  \rho^\mathrm{RHF} \\
&\quad+Cs \|\chi^+_r \rho^\mathrm{RHF} \|^{7/6}_{L^1}\left( \tr (- \Delta \eta_{r} \gamma^\mathrm{RHF} \eta_{r}) \right)^{1/2}.
\end{split}
\end{equation*}
Consequently, we arrive at
\begin{equation*}
\begin{split}
& \left(  \int_{A_r} \rho^\mathrm{RHF}(x) dx \right)^2 \\
&\le
C\sum_{j=1}^K\left( \int_{r \le |x-R_j|<  (1+ \lambda)r}  \rho^\mathrm{RHF}\right)^2\\
&\quad+  C\left(\sup_{x \in A_r} \phi(x)^{-1}[\Phi_{r}^\mathrm{RHF}(x)]_+  + s + (\lambda^2 s)^{-1} +  \lambda^{-1}+  \frac{1}{R_0^2}\right) \int_{A_r}  \rho^\mathrm{RHF} \\
&\quad+Cs \|\chi^+_r \rho^\mathrm{RHF} \|^{7/6}_{L^1}\left( \tr (- \Delta \eta_{r} \gamma^\mathrm{RHF}\eta_{r}) \right)^{1/2}.
\end{split}
\end{equation*}
We now use the fact that  for any $a, c_i, p_i > 0$ if $na^2 \le \sum_{i=1}^n c_i^{p_i} a^{2-p_i}$ then it follows that $a \le \sum_{i=1}^nc_i$ (see the last line in the proof of~\cite[Lemma 7]{mullerIC}).
Then the proof of Lemma \ref{outside} is complete.
\end{proof}


\section{Spliting outside from inside}
Our next task is to extend the conclusion of~\cite[Section 4]{mullerIC}.
We may choose
\[
\eta_-^2 + \eta_+^2 + \eta^2_{r} = 1
\]
with
\[
\mathrm{supp} \, \eta_- \subset A_r^c. \quad \mathrm{supp} \, \eta_+ \subset A_{(1-\lambda)r} \cap A_{(1+\lambda)r}^c,
\]
$\eta_- (x)=1$ if $x \in A_{(1-\lambda)r}^c$, and
\[
\sum_{\# = +, -, r}|\nabla \eta_\#|^2 \le C(\lambda r)^{-2}.
\]
Next, we introduce the screened RHF functional by
\[
\E_{r}^\mathrm{RHF}(\gamma) \coloneqq \tr \left[\left(-\frac{\Delta}{2} - \Phi^\mathrm{RHF}_{r}\right) \gamma \right] + D[\rho_{\gamma}].
\]
In this section, we will prove as follows.
\begin{lemma}
\label{frominside}
For all $r \in (0, R_0], \lambda \in (0, 1/2]$, and for any $0 \le \gamma \le 1$ satisfying
\[
\mathrm{supp} \, \rho_\gamma \subset A_r, \quad \tr \gamma \le \int_{A_r}\rho^\mathrm{RHF},
\]
it holds that
\[
\E_{r}^\mathrm{RHF}\left(\eta_{r} \gamma^\mathrm{RHF} \eta_{r} \right) \le \E_{r}^\mathrm{RHF}(\gamma) + \mathcal{R},
\]
where
\begin{equation}
\begin{split}
\mathcal{R} &\le C\left(1+ (\lambda r)^{-2} \right) \int_{A_{(1-\lambda)r} \cap A_{(1+\lambda)r}^c} \rho^\mathrm{RHF}
+ C\lambda r^3\sup_{x \in A_{(1-\lambda)r}}[\Phi_{(1-\lambda)r}^\mathrm{RHF}(x)]_+^{5/2}.
\end{split}
\end{equation}
\end{lemma}

\begin{proof}
It suffices to show that
\begin{align*}
\E^\mathrm{RHF}(\eta_- \gamma^\mathrm{RHF} \eta_-) +
\E_{r}^\mathrm{RHF}\left(\eta_{r} \gamma^\mathrm{RHF} \eta_{r} \right) - \mathcal{R}
&\le \E^\mathrm{RHF}(\gamma^\mathrm{RHF}) \\
&\le \E^\mathrm{RHF}(\eta_{-} \gamma^\mathrm{RHF} \eta_{-}) + \E_{r}^\mathrm{RHF}(\gamma).
\end{align*}

{\bf Upper bound.} From the minimizing property and the fact that $N \mapsto E^\mathrm{RHF}(N, Z)$ is non-increasing, we have
\[
\E^\mathrm{RHF}(\gamma^\mathrm{RHF}) \le \E^\mathrm{RHF}(\gamma + \eta_- \gamma^\mathrm{RHF} \eta_-)
\]
By direct computation, we have
\begin{equation*}
\begin{split}
 \E^\mathrm{RHF}(\gamma + \eta_- \gamma^\mathrm{RHF} \eta_-)
 &=  \E^\mathrm{RHF}(\eta_- \gamma^\mathrm{RHF} \eta_-) +  \E^\mathrm{RHF}(\gamma) +\iint \frac{\eta_-(x)^2 \rho^\mathrm{RHF}(x) \rho_\gamma(y)}{|x-y|} \, dxdy \\
 &\le \E^\mathrm{RHF}(\eta_- \gamma^\mathrm{RHF} \eta_-) +\E_{r = 0}^\mathrm{RHF}(\gamma) + \iint_{A_r^c}\frac{\rho^\mathrm{RHF}(x) \rho_{\gamma} (y)}{|x-y|}\, dxdy \\
 &= \E^\mathrm{RHF}(\eta_- \gamma^\mathrm{RHF} \eta_-) + \E_{r }^\mathrm{RHF}(\gamma).
\end{split}
\end{equation*}

{\bf Lower bound.} 
By the IMS formula, we have
\begin{equation*}
\begin{split}
\E^\mathrm{RHF}(\gamma^\mathrm{RHF}) &\ge  \E^\mathrm{RHF}(\eta_- \gamma^\mathrm{RHF} \eta_-)  +  \E^\mathrm{RHF}(\eta_+ \gamma^\mathrm{RHF} \eta_+)\\
&\quad +  \E^\mathrm{RHF}(\eta_{r} \gamma^\mathrm{RHF} \eta_{r}) - \sum_{\# = -, +, r} \int |\nabla \eta_{\#}|^2 \rho^\mathrm{RHF} \\
&\quad + \iint \frac{\eta_{r}(x)^2 \rho^\mathrm{RHF}(x) \rho^\mathrm{RHF}(y) (\eta_-(y)^2 + \eta_+(y)^2)}{|x-y|} \, dxdy \\
&\quad +  \iint \frac{\eta_{+}(x)^2 \rho^\mathrm{RHF}(x) \rho^\mathrm{RHF}(y) (\eta_-(y)^2)}{|x-y|} \, dxdy.
\end{split}
\end{equation*}
By construction, we see
\[
- \sum_{\# = -, +, r} \int |\nabla \eta_{\#}|^2 \rho^\mathrm{RHF}
\ge -C(\lambda r)^{-2} \int_{A_{(1-\lambda)r} \cap A_{(1+\lambda)r}^c} \rho^\mathrm{RHF}.
\]
Moreover, we get
\begin{equation*}
\begin{split}
\E^\mathrm{RHF}(\eta_{r} \gamma^\mathrm{RHF} \eta_{r}) 
&+ \iint \frac{\eta_{r}(x)^2 \rho^\mathrm{RHF}(x) \rho^\mathrm{RHF}(y) (\eta_-(y)^2 + \eta_+(y)^2)}{|x-y|} \, dxdy   \\
&\ge
\E^\mathrm{RHF}(\eta_{r} \gamma^\mathrm{RHF} \eta_{r})  
+ \iint  \frac{\eta_{r}(x)^2 \rho^\mathrm{RHF}(x) \rho^\mathrm{RHF}(y) (1-\chi_r^+)}{|x-y|} \, dxdy \\
&\ge \E_{r}^\mathrm{RHF}(\eta_{r} \gamma^\mathrm{RHF} \eta_{r}).
\end{split}
\end{equation*}
Similarly, it follows that
\begin{equation*}
\begin{split}
\E^\mathrm{RHF}(\eta_{+} \gamma^\mathrm{RHF} \eta_{+}) 
&+ \iint \frac{\eta_{+}(x)^2 \rho^\mathrm{RHF}(x) \rho^\mathrm{RHF}(y) \eta_-(y)^2}{|x-y|} \, dxdy \\
&\ge
\E^\mathrm{RHF}(\eta_{+} \gamma^\mathrm{RHF} \eta_{+}) 
+ \iint \frac{\eta_{+}(x)^2 \rho^\mathrm{RHF}(x) \rho^\mathrm{RHF}(y) \1_{A^c_{(1-\lambda)r}}(y)}{|x-y|} \, dxdy \\
&\ge \E_{(1-\lambda)r}^\mathrm{RHF}(\eta_{+} \gamma^\mathrm{RHF} \eta_{+}) \\
&\ge \tr \left[\left(-\frac{\Delta}{2} - \Phi_{(1-\lambda)r} \right) \eta_{+} \gamma^\mathrm{RHF} \eta_{+}\right].
\end{split}
\end{equation*}
Applying Lieb-Thirring inequality with $V = \Phi_{(1-\lambda)r} ^\mathrm{RHF} \1_{\mathrm{supp} \eta_+}$, we see that
\begin{equation*}
\begin{split}
\tr \left[\left(-\frac{\Delta}{2} - \Phi_{(1-\lambda)r}^\mathrm{RHF}  \right) \eta_{+} \gamma^\mathrm{RHF} \eta_{+}\right]
&\ge \tr \left( - \frac{\Delta}{2} - V\right)_- \\
&\ge -C\sum_{j=1}^K \int_{{(1-\lambda)r} \le |x-R_j|\le(1+\lambda)r} \left[\Phi_{(1-\lambda)r} ^\mathrm{RHF} \right]_+^{5/2} \\
&\ge -C\lambda r^3\sup_{x \in A_{(1-\lambda)r}}[\Phi_{(1-\lambda)r}^\mathrm{RHF}(x)]_+^{5/2}.
\end{split}
\end{equation*}
Hence
\begin{equation*}
\begin{split}
\E^\mathrm{RHF}(\gamma^\mathrm{RHF}) &\ge  \E^\mathrm{RHF}(\eta_- \gamma^\mathrm{RHF} \eta_-) +
\E_{r}^\mathrm{RHF}(\eta_{r} \gamma^\mathrm{RHF} \eta_{r}) \\
&-C(1+\lambda r)^{-2} \int_{A_{(1-\lambda)r} \cap A_{(1+\lambda)r}^c} \rho^\mathrm{RHF} \\
&-C\lambda r^3 \sup_{x \in A_{(1-\lambda)r}}[\Phi_{(1-\lambda)r}^\mathrm{RHF}(x)]_+^{5/2}.
\end{split}
\end{equation*}
This completes the proof.
\end{proof}

By pursuing the above reasoning, one can show the following lemma. 
\begin{lemma}
\label{inside2}
For any $r \in (0, R_0]$ and any $\lambda \in (0, 1/2]$ we have
\begin{equation}
\begin{split}
\tr \left(-\frac{\Delta}{2} \eta_r \gamma^\mathrm{RHF} \eta_r\right) &\le C(1+(\lambda r)^{-2})\int_{A_{(1-\lambda)r} } \rho^\mathrm{RHF} + C\lambda r^3\sup_{x \in A_{(1-\lambda)r}}[\Phi_{(1-\lambda)r}^\mathrm{RHF}(x)]_+^{5/2}\\
&\quad  + C\sup_{x \in A_r}[\phi(x)^{-1}\Phi_r^\mathrm{RHF}(x)]_+^{7/3}.
\end{split}
\end{equation}
\end{lemma}

\begin{proof}
We apply Lemma \ref{frominside} with $\gamma = 0$ and obtain $\E_r^\mathrm{RHF} (\eta_r \gamma^\mathrm{RHF} \eta_r)\le \mathcal{R}$.
On the other hand, by the kinetic Lieb-Thirring inequality and the fact that the ground state energy in Thomas-Fermi theory is  $-\mathrm{const.}\sum_{j=1}^Kz_j^{7/3}$~\cite{LiebTF, LiebSimon1977}, we have

\begin{align*}
\E_r^\mathrm{RHF} (\eta_r \gamma^\mathrm{RHF} \eta_r) 
&\ge  \tr \left(-\frac{\Delta}{4} \eta_r \gamma^\mathrm{RHF} \eta_r\right) +C^{-1} \int(\eta_r^2 \rho^\mathrm{RHF})^{5/3}  \\
&\quad- \sup_{x \in A_r} \phi(x)^{-1}[\Phi_{r}^\mathrm{RHF}(x)]_+\sum_{j=1}^K\int\eta_r^2\frac{z_j\rho^\mathrm{RHF}(x)}{Z|x-R_j|}\, dx +D[\eta_r^2 \rho^\mathrm{RHF}] \\
&\ge  \tr \left(-\frac{\Delta}{4} \eta_r \gamma^\mathrm{RHF} \eta_r\right)  - C \sup_{x \in A_r}[\phi(x)^{-1}\Phi_{r}^\mathrm{RHF}(x)]_+^{7/3}.
\end{align*}
Therefore,
\[
\tr \left(-\frac{\Delta}{2} \eta_r \gamma^\mathrm{RHF} \eta_r\right) \le C\mathcal{R} +C\sup_{x \in A_r}[\phi(x)^{-1}\Phi_{r}^\mathrm{RHF}(x)]_+^{7/3}
\]
which implies the conclusion.
\end{proof}
\section{Sommerfeld estimates}
In this section, we will show the Sommerfeld asymptotics for molecules.
Let $\Gamma_j$ be the Voronoi cells $\Gamma_j \coloneqq \{ x \in \R^3 \colon |x-R_j| < |x-R_i| \, \text{ for all } i\neq j\}$, for $j=1, \dots, K$.
The following theorem is a generalization of~\cite[Theorem 4.6]{SolovejIC} and~\cite[Lemma 3.11]{Samojlow}.
\begin{theorem}[Sommerfeld asymptotics]
\label{Sommerfeld1}
Let $r \in (0, R_0]$ and $\phi$ be the TF potential satisfying $\Delta\phi = 4\pi c_\mathrm{TF}^{-3/2}[\phi -\mu]_+^{3/2}$ in $A_r$, where $c_\mathrm{TF} = 2^{-1}(3\pi^2)^{2/3}$, and $\mu \ge 0$ is a chemical potential.
We assume $\lim_{s\to +r} \inf_{\partial A_s} \phi > \mu$, and $\phi$ is continuous on $A_r$ and vanishes at infinity.
Then for any $x \in A_r$ it follows that
\[
\max\left\{ \max_{1 \le j \le K }\omega_a^-(x-R_j), \max_{1 \le j \le K }\frac{\nu(\mu, r)}{|x-R_j|} \right\}\le \phi(x) \le \sum_{j=1}^K\omega_A^+(x-R_j) + \mu,
\]
where $\nu(\mu, r) \coloneqq \inf_{|x|\ge r} \max \{ \mu|x|, \omega_a^-(x)|x|\}$ and
\[
a(r) \coloneqq \liminf_{s \to +r} \sup_{\partial A_s} \left(\sqrt{c_\mathrm{S} r^{-4}\phi^{-1}} -1\right), \quad \omega_a^-(x) \coloneqq c_\mathrm{S}|x|^{-4}\left(1+a(r)\left(r|x|^{-1} \right)^\xi \right)^{-2},
\]
\[
A(r) \coloneqq \liminf_{s \to +r} \sup_{\partial A_s} \left(c_\mathrm{S}^{-1} s^{4}(\phi-\mu) -1\right), \quad \omega_A^+(x) \coloneqq c_\mathrm{S}|x|^{-4}\left(1+A(r)\left(r|x|^{-1} \right)^\xi \right).
\]
Here $\xi = (-7+\sqrt{73})/2 \sim 0.77$ and $c_\mathrm{s} = 3^42^{-3}\pi^2$.
\end{theorem}

\begin{proof}
\textbf{Step 1}
By assumption, there is a $r_0 \in (r, R_0)$ such that $\inf_{\partial A_s} \phi > \mu \ge0$ for any $s \in (r, r_0)$.
Hence $a(s)$ is well-defined for any $s \in (r, r_0)$.
We prove the claim with $r$ replaced by arbitrary $s \in (r, r_0)$ and take the limit $s \to r$.

{\bf Step 2} (Lower bound)
We consider $f(x) \coloneqq \max \{ \max_{1 \le j \le K }\omega_a^-(x-R_j),  \nu \max_{1 \le j \le K } |x-R_j|^{-1}\}$ on $\overline{A_s}$.
Since $\inf_{\partial A_s} \phi > \mu$, we have $a(s) > -1$.
By definition, we have
\begin{itemize}
\item[(a)]
$\omega_a^- (x)|x|$ is positive and radial for $|x| \ge s$.

\item[(b)]
$\omega_a^-(x) = \inf_{\partial A_s} \phi > \mu$ for any $|x| = s$.

\item[(c)]
$\Delta \omega_a^- (x) \ge 4\pi c_\mathrm{TF}^{-3/2}\omega_a^-(x)^{3/2}$ for any $|x|>s$.
\end{itemize}
Indeed, (a) and (b) are followed by definition.
Then (c) is obtained in~\cite[Eq. (38)]{SolovejIC}.
From (a), (b), and the fact that $\mu |x|$ is increasing, there is a $R \in (s, \infty)$ so that $\omega_a^-(|x|=R) = \mu$ and $\nu = \mu R$.
Moreover, for any $x \in \overline{A_s}$
\begin{equation}
   f(x) 
  = \begin{cases}
       \max_{1 \le j \le K }\omega_a^-(x-R_j)  & \text{if $f(x) > \mu$} \\
     \nu \max_{1 \le j \le K } |x-R_j|^{-1} & \text{if $f(x) \le \mu$}.
    \end{cases}
\end{equation}
Thus, by (b) we have $f|_{\partial A_s} = \omega_a^-(|x-R_j|=s) = \inf_{\partial A_s} \phi$.
Let $u \coloneqq f - \phi$.
It suffices to show that $\Delta u \ge 0$ in $A_s \cap \{ u > 0\}$.
From $\Delta u = \Delta  f - 4\pi c_\mathrm{TF}^{-3/2} [\phi - \mu]_+^{3/2}$ we will show that
\[
\Delta f \ge 4\pi c_\mathrm{TF}^{-3/2} [f-\mu]_+^{3/2} \quad \text{in } A_s.
\]
For any nonnegative function $\psi \in C_c^\infty(A_r\cap \{f > \mu\})$ we may compute
\begin{align*}
\int_{\R^3} f \Delta \psi &= \sum_{j=1}^K\int_{\Gamma_j} \mathrm{div}( \omega_a^-(x-R_j) \nabla \psi(x)) \, dx - \sum_{j=1}^K\int_{\Gamma_j} \nabla  \omega_a^-(x-R_j) \cdot \nabla \psi(x) \, dx \\
&=  \sum_{j=1}^K\int_{\partial \Gamma_j}  \omega_a^-(x-R_j) n_j \cdot \nabla \psi(x) \, dS -\sum_{j=1}^K\int_{\Gamma_j} \nabla  \omega_a^-(x-R_j) \cdot \nabla \psi(x) \,dx
\end{align*}
by Gauss's theorem.
Here $n_j$ is the outward normal of $\partial \Gamma_j$.
We note that the first integral is zero by the fact that $n_j = - n_k$ on $\partial \Gamma_j \cap \partial \Gamma_k$.
Similarly,
\begin{align*}
- \sum_{j=1}^K&\int_{\Gamma_j} \nabla  \omega_a^-(x-R_j) \cdot \nabla \psi(x) \, dx \\
&=  -\sum_{j=1}^K\int_{\Gamma_j} \mathrm{div}( \psi(x) \nabla  \omega_a^-(x-R_j)) \, dx + \sum_{j=1}^K\int_{\Gamma_j} \psi(x) \Delta \omega_a^-(x-R_j) \,dx  \\
&\ge -\sum_{j=1}^K\int_{\partial \Gamma_j}  \psi(x) n_j \cdot \nabla \omega_a^-(x-R_j) \, dS +4 \pi c_\mathrm{TF}^{-3/2} \int \psi f(x)^{3/2} \, dx.
\end{align*}
From the fact  that $n_j \cdot \nabla \omega_a^-(x-R_j) \le 0$ on $\partial \Gamma_j$ (because $\Gamma_j$ is convex),  we have
\[
\int_{\R^3} f \Delta \psi  \ge 4 \pi c_\mathrm{TF}^{-3/2} \int \psi [f-\mu]_+^{3/2},
\]
and thus $\Delta f \ge 4 \pi c_\mathrm{TF}^{-3/2} [f-\mu]_+^{3/2}$ in $A_s \cap \{f > \mu\}$.
We note $\omega_a^-$ is subharmonic and $|x-R_j|^{-1}$ is harmonic on $A_s$.
Thus $\Delta f \ge 0$ in $A_s$.
We pick any nonnegative function $\psi \in C_c^\infty(A_s)$ and a nonnegative monotone sequence $0 \le \xi_n \in C_c^\infty(\{f > \mu\})$ so that $\xi_n \to \1_{\{f>\mu\}}$ pointwise in $\mathrm{supp } \, \psi$.
Then, with the above results, we find
\[
\int f\Delta\psi = \int f \Delta(\xi_n \psi) + \int f \Delta(1-\xi_n)\psi
\ge 4 \pi c_\mathrm{TF}^{-3/2} \int [f-\mu]_+^{3/2} \xi_n\psi \to 4 \pi c_\mathrm{TF}^{-3/2} \int [f-\mu]_+^{3/2}\psi 
\]
by monotone convergence theorem.
Hence $\Delta u \ge 0$ in $A_s \cap \{u > 0\}$ holds.
From the maximum principle, $A_s \cap \{u > 0\} $ is empty.
Hence $f \le \phi$ follows.

{\bf Step 3} (upper bound)
We consider $g(x) \coloneqq \sum_{j=1}^K\omega_A^+(x-R_j) + \mu$.
Since $\Delta \omega_A^+ \le 4 \pi c_\mathrm{TF}^{-3/2}(\omega_A^+)^{3/2}$ in $|x| \ge s$,
it satisfies that $\Delta g \le 4 \pi c_\mathrm{TF}^{-3/2}[g-\mu]_+^{3/2}$ in $A_s$.
By $\omega_j^+|_{\partial A_s} = \sup_{\partial A_s} \phi -\mu$,  we have $g(x) \ge \omega_A^+(|x-R_j| = s)  + \mu =  \sup_{\partial A_s} \phi $ for any $x \in \partial A_s$.
Let $u \coloneqq \phi - g$.
Then we have, on $g < \phi$,
\[
\Delta u \ge  4 \pi c_\mathrm{TF}^{-3/2}([\phi - \mu]_+^{3/2} - [g - \mu]^{3/2}_+) \ge 0.
\]
Hence we learn $\phi \le g$ on $A_s$ by the maximum principle.
\end{proof}

Next, as in~\cite[Lemma 3.12]{Samojlow}, we improve the upper bound for $x$ close to $\partial A_r$. Namely, we will show the following theorem.
\begin{theorem}[Refined upper bound]
\label{Sommerfeld2}
Let $r \in (0, R_0]$, $\mu \ge 0$, and $\phi$ is continuous on $A_r$ and vanishes at infinity.
We assume $\Delta \phi = 4\pi c_\mathrm{TF}^{-3/2}[\phi-\mu]_+^{3/2}$ in $A_r$.
Then it holds that, for $j=1,\dots K$,
\[
\phi(x) \le \omega_{A_1, A_2}^j(x-R_j) + \mu \quad \text{if } x \in A_r \cap \Gamma_j,
\]
 where
\begin{align*}
 \omega_{A_1, A_2}^j(x) &\coloneqq c_\mathrm{S}|x|^{-4}\left(1 + A_1^j(r) \left(\frac{|x|}{R_j} \right)^\eta + A_2^j(r)\left(\frac{r}{|x|} \right)^\xi \right), \\
 R_j &\coloneqq \frac{1}{2}\min_{i \neq j}|R_i - R_j|, \quad A_i^j(r) \coloneqq \liminf_{s \to +r} B_i^j(s),\quad i= 1,2, \\
 B_1^j(s) &\coloneqq \frac{4+B_2^j(s)(4+\xi)\left(\frac{s}{R_j} \right)^\xi}{\eta -4}, \\
  B_2^j(s) &\coloneqq \left[ \frac{\sup_{\partial A_s} \left(c_\mathrm{S}^{-1}s^4(\phi-\mu) -1 \right) - \frac{4}{\eta-4}\left(\frac{s}{R_j} \right)^\eta}{1+\frac{4+\xi}{\eta-4} \left(\frac{s}{R_j} \right)^{\xi + \eta}} \right]_+.
\end{align*}
Here $\eta = (7+\sqrt{73})/2 \sim 7.772$.
\end{theorem}

\begin{proof}
We prove the upper bound with $r$ replaced by any $s \in (r, R_0)$.
Then $A_i^j(s) = B_i^j(s)$ for $i=1, 2$.
Our strategy is to apply the maximum principle to the function
\[
u(x) \coloneqq \phi(x) -\left(\sum_{j=1}^K\omega_{B_1, B_2}^j(x-R_j)\1_{\Gamma_j}(x) + \mu\right).
\]
By definition, we have $u(x) \le 0$ on $\partial A_s$.
Hence it suffices to show that $-\Delta u\le 0$ in $A_s \cap \{u>0 \}$.

For any nonnegative function $\psi \in C_c^\infty(A_s \cap \{u>0\})$ we may compute
\begin{align*}
\int_{\R^3} u(x) \Delta \psi(x) \, dx &= \int_{\R^3} \phi(x) \Delta \psi (x) \, dx - \sum_{j=1}^K \int_{\Gamma_j} \omega_{B_1, B_2}^j(x-R_j)\Delta \psi(x) \, dx.
\end{align*}
The second integral is
\begin{align*}
\sum_{j=1}^K \int_{\Gamma_j} \omega_{B_1, B_2}^j(x-R_j)\Delta \psi(x) \, dx 
&=  \sum_{j=1}^K \int_{\partial \Gamma_j} \omega_{B_1, B_2}^j(x-R_j)n_j \cdot \nabla \psi(x) \, dx \\
&\quad - \sum_{j=1}^K \int_{\Gamma_j} \nabla\omega_{B_1, B_2}^j(x-R_j) \cdot \nabla \psi(x) \, dx,
\end{align*}
by Gauss's theorem.
The first integral is zero from the continuity.
We note that $\Delta \omega_{B_1, B_2}^j \le 4\pi c_\mathrm{TF}^{-3/2}(\omega_{B_1, B_2}^j)^{3/2}$ for $|x| \neq 0$.
Then we have
\begin{align*}
\int_{\R^3} u(x) \Delta \psi(x) \, dx &\ge  \sum_{j=1}^K \int_{\partial \Gamma_j} \psi(x)n_j \cdot \nabla  \omega_{B_1, B_2}^j(x-R_j) \, dx.
\end{align*}
By direct computation, we see
\[
\nabla  \omega_{B_1, B_2}^j(x) = c_\mathrm{S}\frac{x}{|x|^6}\left(B_1^j(\eta - 4)\left(\frac{|x|}{R_j}\right)^\eta - B_2^j(4 + \xi)\left(\frac{r}{|x|}\right)^\xi -4 \right).
\]
From the convexity of $\Gamma_j$, we learn $n_j \cdot (x-R_j) \ge 0$ on $\partial \Gamma_j$.
Hence $n_j \cdot \nabla  \omega_{B_1, B_2}^j(x-R_j) \ge 0$.
This shows $\Delta u \ge 0$.
\end{proof}


\section{Initial step}
From now on, we assume $N \ge Z \ge 1$.
In this section, our goal is as follows.
\begin{lemma}[initial step]
\label{initial}
There is a universal constant $C_1 > 0$ so that
\begin{equation}
\label{initial1}
\sup_{x \in \partial A_r}\left|\Phi^\mathrm{RHF}_{r}(x) - \Phi_{r}^\mathrm{TF}(x) \right| \le C_1Z^{49/36 - a}r^{1/12},
\end{equation}
for all $r \in (0, R_0]$ with $a=1/198$.
\end{lemma}

\begin{proof}
The strategy is to bound $\E^\mathrm{RHF}(\gamma^\mathrm{RHF})$ from above and below by using the semi-classical estimates.

{\bf Upper bound.}
We will show that
\begin{equation}
\label{upper}
\E^\mathrm{RHF}(\gamma^\mathrm{RHF}) \le \E^\mathrm{TF}(\rho^\mathrm{TF}) + CZ^{25/11}.
\end{equation}

Since $E^\mathrm{RHF}(N, Z)$ is non-increasing in $N$ we have
\[
\E^\mathrm{RHF}(\gamma^\mathrm{RHF}) \le \inf\{ \E^\mathrm{RHF}(\gamma) \colon 0 \le \gamma \le 1, \, \tr \gamma \le N\}.
\]

We now use the following lemma taken from \cite[Lemma 11] {mullerIC} and \cite[Lemma 8.2]{SolovejIC}.
\begin{lemma}
\label{semicl}
For fixed $s>0$ and smooth $g \colon \R^3 \to [0, 1]$ satisfying $\mathrm{supp} \, g \subset \{|x|<s\}$, $\int g^2=1$, $\int |\nabla g|^2 \le Cs^{-2}$ it follows that

\begin{enumerate}
\item[(1)] For any $V \colon \R^3 \to \R$ with $[V]_+$, $[V - V \star g^2]_+ \in L^{5/2}$ and for any $0 \le \gamma \le 1$
\begin{equation*}
\begin{split}
\tr \left[\left(- \frac{\Delta}{2} -V\right)\gamma\right] &\ge -2^{5/2}(15 \pi^2)^{-1}\int[V]^{5/2}_+ - Cs^{-2} \tr \gamma \\
& \quad -C\left( \int [V]^{5/2}_+ \right)^{3/5} \left(\int [V-V \star g^2]^{5/2}_+\right)^{2/5}.
\end{split}
\end{equation*}

\item[(2)]
If $[V]_+ \in L^{5/2} \cap L^{3/2}$, then there is a density-matrix $\gamma$ so that $\rho_\gamma = 2^{5/2}(6\pi^2)^{-1}[V]^{3/2}_+ \star g^2$,
\[
\tr \left( -\frac{\Delta}{2}\gamma \right) \le 2^{3/2}(5\pi^2)^{-1} \int [V]_+^{5/2}
+ Cs^{-2}\int[V]_+^{3/2}.
\]
\end{enumerate}
\end{lemma}

We introduce the Thomas-Fermi potential
\[
\phi^\mathrm{TF}(x) = V_Z(x) - \rho^\mathrm{TF} \star |x|^{-1}
\]
and apply Lemma \ref{semicl} (2) with $V=\phi^\mathrm{TF}$ and a spherically symmetric $g$ to obtain a density matrix $\gamma'$.
Because of the Thomas-Fermi equation we have
\[
\rho_{\gamma'} = 2^{5/2}(6\pi^2)^{-1}(\phi^\mathrm{TF})^{3/2}\star g^2 = \rho^\mathrm{TF} \star g^2.
\]
Since
\[
\tr \gamma' = \int \rho_{\gamma'} = \int \rho^\mathrm{TF} =Z\le N,
\]
we obtain
\[
\inf\{\E^\mathrm{RHF}(\gamma) \colon 0 \le \gamma \le 1, \, \tr \gamma \le N\} \le \E^\mathrm{RHF}(\gamma').
\]
Again by Lemma \ref{semicl} (2), we have
\begin{equation*}
\label{semibound}
\begin{split}
\E^\mathrm{RHF}(\gamma') &\le 2^{3/2}(5\pi^2)^{-1}\int [V]_+^{5/2} + Cs^{-2}\int[V]_+^{3/2} -\int V_Z(\rho^\mathrm{TF}\star g^2) +D[\rho^\mathrm{TF}\star g^2]\\
&\le \frac{3}{10}c_\mathrm{TF}\int_{\R^3}\rho^\mathrm{TF}(x)^{5/3}\, dx - \int V_Z\rho^\mathrm{TF} + D[\rho^\mathrm{TF}] \\
&+ Cs^{-2}\int \rho^\mathrm{TF} +\int(V_Z - V_Z\star g^2)\rho^\mathrm{TF}  \\
&=\E^\mathrm{TF}(\rho^\mathrm{TF})+ Cs^{-2}\int \rho^\mathrm{TF} +\int(V_Z - V_Z\star g^2)\rho^\mathrm{TF}.
\end{split}
\end{equation*}
In the second inequality, we have used $[g^2 \star |x|^{-1} \star g^2](x-y) \le |x-y|^{-1}$.
This fact follows from Fourier transform.
By Newton's theorem, we see
\begin{equation}
\label{Newton}
V_Z - V_Z\star g^2 = \sum_{j=1}^Kz_j\left(|x-R_j|^{-1} \1(|x-R_j| \le s)\right).
\end{equation}
Then, by H\"{o}lder's inequality,
\begin{equation*}
\begin{split}
\int(V_Z - V_Z\star g^2)\rho^\mathrm{TF}
&\le \left( \int_{\R^3}\rho^\mathrm{TF}(x)^{5/3} \, dx\right)^{3/5}\left( \int (V_Z - V_Z\star g^2)^{5/2}\right)^{2/5} \\
&\le CZ^{12/5}\left(\sum_{i=1}^K z_i/Z\int_{|x-R_i|\le s} |x-R_i|^{-5/2}\right)^{2/5}\, dx \\
&\le CZ^{12/5} s^{1/5},
\end{split}
\end{equation*}
where we have used (\ref{Newton}) and the convexity of $x^{5/2}$.
Thus, after optimization in $s$, we get
\[
\E^\mathrm{RHF}(\gamma') \le \E^\mathrm{TF}(\rho^\mathrm{TF}) + CZ^{25/11}.
\]
This shows the desired upper bound.

{\bf Lower bound.}
We will show that
\begin{equation}
\label{lower}
\E^\mathrm{RHF}(\gamma^\mathrm{RHF}) \ge \E^\mathrm{TF}(\rho^\mathrm{TF}) + D[\rho^\mathrm{RHF} - \rho^\mathrm{TF}] -CZ^{25/11}.
\end{equation}
We can write
\[
\E^\mathrm{RHF}(\gamma^\mathrm{RHF})  = \tr \left[\left(-\frac{\Delta}{2} - \phi^\mathrm{TF}\right)\gamma^\mathrm{RHF}\right] + D[\rho^\mathrm{RHF} - \rho^\mathrm{TF}] - D[\rho^\mathrm{TF}].
\]
Then, from Lemma \ref{semicl} (1) we have
\begin{equation*}
\begin{split}
\tr \left[\left(-\frac{\Delta}{2} - \phi^\mathrm{TF}\right)\gamma^\mathrm{RHF} \right]
&\ge -2^{5/2}(15 \pi^2)^{-1} \int_{\R^3}\phi^\mathrm{TF}(x)^{5/2}\, dx -Cs^{-2} \tr \gamma^\mathrm{RHF} \\
&-C\left( \int_{\R^3} \phi^\mathrm{TF}(x)^{5/2} \, dx \right)^{3/5} \left(\int [\phi^\mathrm{TF}-\phi^\mathrm{TF} \star g^2]_+^{5/2}\right)^{2/5}.
\end{split}
\end{equation*}
By the TF equation, we see that
\[
\int_{\R^3} \phi^\mathrm{TF}(x)^{5/2} \, dx = C \int_{\R^3}\rho^\mathrm{TF}(x)^{5/3} \le C{Z^{7/3}}.
\]
Since $V_Z - V_Z \star g^2 \ge 0$, because $V_Z$ is superharmonic, we obtain
\[
\int [\phi^\mathrm{TF}-\phi^\mathrm{TF} \star g^2]_+^{5/2} 
\le \int [V_Z- V_Z \star g^2]_+^{5/2}
\le CZ^{5/2}s^{1/2}.
\]
Hence we find that
\begin{equation*}
\begin{split}
\tr \left[\left(-\frac{\Delta}{2} - \phi^\mathrm{TF}\right)\gamma^\mathrm{RHF} \right]
&\ge -2^{5/2}(15 \pi^2)^{-1}\int_{\R^3}\phi^\mathrm{TF}(x)^{5/2}\, dx -Cs^{-2}Z - CZ^{12/5}s^{1/5}.
\end{split}
\end{equation*}
Optimizing over $s>0$, we get
\[
\tr \left[\left(-\frac{\Delta}{2} - \phi^\mathrm{TF}\right)\gamma^\mathrm{RHF} \right]
\ge -2^{5/2}(15 \pi^2)^{-1}\int_{\R^3}\phi^\mathrm{TF}(x)^{5/2}\, dx -CZ^{25/11}.
\]
Using the relation from the TF equation
\[
-2^{5/2}(15 \pi^2)^{-1} \int_{\R^3}\phi^\mathrm{TF}(x)^{5/2}\, dx -D[\rho^\mathrm{TF}] = \E^\mathrm{TF}(\rho^\mathrm{TF}),
\]
we arrive at the lower bound (\ref{lower}).

{\bf Conclusion.}
Combining (\ref{upper}) and ($\ref{lower}$), we infer that
\begin{equation}
D[\rho^\mathrm{RHF}-\rho^\mathrm{TF}] \le CZ^{25/11}.
\end{equation}

The following lemma is taken from~\cite[Cor. 9.3]{SolovejIC} and \cite[Lemma 12]{mullerIC}.
\begin{lemma}[Coulomb estimate]
\label{coulomb}
For every $f \in L^{5/3}(\R^3) \cap L^{6/5}(\R^3)$ and $x \in \R^3$, we have
\begin{equation*}
\left| \int_{|y| < |x|} \frac{f(y)}{|x-y|} \, dy \right| \le C \| f \|_{L^{5/3}}^{5/6}(|x|D[f])^{1/12}.
\end{equation*}
\end{lemma}
Using this Coulomb estimate with $f(y) = (\rho^\mathrm{RHF} - \rho^\mathrm{TF})(y+R_j)$, we find that, for $r \in (0, R_0]$,
\begin{equation}
\begin{split}
\sup_{x \in A_r}|\Phi^\mathrm{RHF}_{r}(x) - \Phi_{r}^\mathrm{TF}(x)| &\le \sum_{j=1}^K\sup_{|x-R_j| = r}\left| \int_{|y| < r}\frac{\rho^\mathrm{RHF}(y+R_j)-\rho^\mathrm{TF}(y+R_j)}{|x-R_j -y|} \, dy \right| \\
&\le C \|\rho^\mathrm{RHF} - \rho^\mathrm{TF}\|_{L^{5/3}}^{5/6}(rD[\rho^\mathrm{RHF}-\rho^\mathrm{TF}])^{1/12} \\
&\le C \|\rho^\mathrm{RHF} - \rho^\mathrm{TF}\|_{L^{5/3}}^{5/6}r^{1/12}Z^{25/132},
\end{split}
\end{equation}
where we have used the harmonicity.
Combining this with the kinetic energy estimates
\[
\int (\rho^\mathrm{RHF})^{5/3} \le CZ^{7/3}, \quad \int (\rho^\mathrm{TF})^{5/3} \le CZ^{7/3},
\]
we find that
\[
\sup_{x \in \partial A_r}|\Phi^\mathrm{RHF}_{r}(x) - \Phi_{r}^\mathrm{TF}(x)| \le CZ^{179/132}r^{1/12},
\]
for all $r \in (0, R_0]$.
Since $179/132 = 49/36 - 1/198$, this implies the desired bound (\ref{initial1}).
\end{proof}


\section{Iterative step}
In this section, we will prove the following theorem.
\begin{theorem}[iterative step]
\label{Thm.iterative}
There are universal constants $C_2, \beta, \delta, \epsilon >0$ such that, if 
\begin{align}
\label{initialassumption}
\sup_{x \in \partial A_s}\left|\Phi^\mathrm{RHF}_{s}(x) - \Phi_{s}^\mathrm{TF}(x) \right| \le \beta s^{-4} \quad \text{for any } s \le D,
\end{align}
where $D \in [Z^{-1/3}, R_0]$, then, with $r \coloneqq D^{1+\delta}$ and $\tilde r \coloneqq R_0^{-1}r^{\frac{\xi}{\xi + \eta}} R_\mathrm{min}^{\frac{\eta}{\xi + \eta}}$, it follows that
\begin{equation}
\sup_{x \in \partial A_s}\left| \Phi^\mathrm{RHF}_{s}(x) - \Phi_{s}^\mathrm{TF}(x)\right| \le C_2s^{-4+\epsilon} \quad \text{for any } s\in \left[r^{\frac{1}{1+\delta}},  \min\{r^{\frac{1-\delta}{1+\delta}}, \tilde r\}\right].
\end{equation}
\end{theorem}

\textit{Step 1}
We collect some consequences of (\ref{initialassumption}).
\begin{lemma}
\label{lem.ite}
We assume that  (\ref{initialassumption}) holds for some $\beta, D \in (0, R_0]$.
Then, if $r \in (0, D]$, we have

\begin{equation}
\label{ite2}
\sup_{x \in A_r} \phi(x)^{-1}[\Phi_{r}^\mathrm{RHF}(x)]_+ \le \frac{C}{r^{3}} ,
\end{equation}

\begin{equation}
\label{ite1}
\left| \sum_{j=1}^K \int_{ |x-R_j| < r} (\rho^\mathrm{RHF} - \rho^\mathrm{TF}) \right| \le \frac{C\beta}{r^3},
\end{equation}

\begin{equation}
\label{ite3}
\int_{A_r} \rho^\mathrm{RHF} \le \frac{C}{r^3},
\end{equation}

\begin{equation}
\label{ite4}
\int_{A_r} (\rho^\mathrm{RHF})^{5/3} \le \frac{C}{r^7},
\end{equation}

\begin{equation}
\label{ite5}
\tr (-\Delta \eta_r \gamma^\mathrm{RHF} \eta_r) \le C\left(\frac{1}{r^7}+ \frac{1}{\lambda^2 r^5}\right), \quad \text{for any } \lambda \in (0, 1/2].
\end{equation}
\end{lemma}

\begin{proof}
First, we  split
\begin{equation*}
\begin{split}
\Phi^\mathrm{RHF}_{r}(x) &= \Phi^\mathrm{RHF}_{r}(x)- \Phi_{r}^\mathrm{TF}(x) + \Phi^\mathrm{TF}_{r}(x).
\end{split}
\end{equation*}
Moreover, we may write
\begin{align*}
\Phi^\mathrm{TF}_{r}(x) &= \phi^\mathrm{TF}(x) + \int \frac{\rho^\mathrm{TF}(y)}{|x-y|} \, dy - \sum_{j=1}^K\int_{|y-R_j| < r}\frac{\rho^\mathrm{TF}(y)}{|x-y|} \, dy \\
&= \phi^\mathrm{TF}(x) +\int_{A_r}\frac{\rho^\mathrm{TF}(y)}{|x-y|} \, dy.
\end{align*}
Using the Sommerfeld bound $\phi^\mathrm{TF}(x) \le c|x-R_j|^{-4}$ on $A_r\cap \Gamma_j$  and the TF equation $c_\mathrm{TF}\rho^\mathrm{TF}(x)^{2/3} =\phi^\mathrm{TF}(x)$, we have
\begin{equation*}
\begin{split}
\phi^\mathrm{TF}(x) + \int_{A_r}\frac{\rho^\mathrm{TF}(y)}{|x-y|} \, dy 
\le C \sum_{j=1}^K\left(|x-R_j|^{-4} +  \int_{|y| >s}\frac{dy}{|x-R_j-y||y|^6} \right)  \le Cr^{-4},
\end{split}
\end{equation*}
for $x \in A_r$, where we have used Newton's theorem.
Hence, by assumption (\ref{initialassumption}), it holds that $\left| \Phi^\mathrm{RHF}_{r}(x) \right| \le Cr^{-4}$ for any $x \in \partial A_r.$
We note that $-\Delta  \Phi^\mathrm{RHF}_{r}(x) = 4\pi \1_{A_r^c}(x) \rho^\mathrm{HF}(x)$ in the distributional sense, and hence $\Phi^\mathrm{RHF}_{r}$ is harmonic in $A_r$.
As in \cite[Lemma 6.5]{TFDWIC}, we may show the following lemma.
\begin{lemma}
\label{harmonic}
Let $f \colon A_r \to \R$ and $g  \colon A_r \to \R_+$.
We assume that $f, g$ are harmonic and continuous  in $A_r$ and vanishing at infinity.
If $g(x) \ge C_0^{-1}r^{-1}$ on $\partial A_r$, then it holds that
\[
\sup_{x\in A_r}g(x)^{-1} f(x) \le C_0r\sup_{x \in \partial A_r} f(x).
\]
\end{lemma}

\begin{proof}
Let $h(x) \coloneqq f(x) - F_rg(x)$ with $F_r =C_0r  \sup_{z \in \partial A_r}f(z)$.
Since $f, g$ are harmonic in $A_r$, by the maximum principle, we have
\[
\sup_{x \in A_r} h(x) = \max \left\{\sup_{x \in \partial A_r} (f(x) - F_r g(x)), 0\right\} = 0
\]
Therefore, for any $x \in A_r$ we learn
\[
 f(x)g(x)^{-1}  = h(x)g(x)^{-1} + F_r \le F_r,
\]
and thus the lemma follows.
\end{proof}

Now we apply this lemma with $f = [\Phi_r^\mathrm{RHF}]_+$ and $g(x) = \phi(x)$. We note that $\phi(x) \ge C^{-1} r^{-1}$ on $\partial A_r$, where $C$ is independent of $Z$ (recall our assumption of Theorem~\ref{theorem.main}). 
Then we have
\[
\sup_{x \in A_r}\phi(x)^{-1}[\Phi_r^\mathrm{RHF}(x)]_+ \le Cr\sup_{x \in \partial A_r}[\Phi_r^\mathrm{RHF}(x)]_+\le Cr^{-3},
\]
which proves (\ref{ite2}).

Next, we note that
\[
\sum_{j=1}^K\int_{|y-R_j| < r} (\rho^\mathrm{TF}(y) - \rho^\mathrm{RHF}(y)) \, dy = \lim_{ |x|\to \infty} \phi(x)^{-1}\left(\int_{A_r^c}\frac{\rho^\mathrm{TF}(y) - \rho^\mathrm{RHF}(y)}{|x-y|} \, dy \right).
\]
Then  (\ref{ite1}) follows from Lemma~\ref{harmonic} and (\ref{initialassumption}).

Now we prove (\ref{ite3}) and (\ref{ite5}). By (\ref{ite1}), we have
\begin{equation*}
\begin{split}
\int_{A_{r/3}\cap A_r^c} \rho^\mathrm{RHF} (x) \, dx &= \int_{A_r^c} (\rho^\mathrm{RHF}(x) - \rho^\mathrm{TF}(x))\, dx - \int_{A_{r/3}^c} (\rho^\mathrm{RHF}(x) - \rho^\mathrm{TF}(x))\, dx\\
&\quad+ \sum_{j=1}^K\int_{3/r \le |x-R_j| \le r} \rho^\mathrm{TF}(x)  \,dx\\
& \le Cr^{-3},
\end{split}
\end{equation*}
where we have used  the Sommerfeld asymptotics $\rho^\mathrm{TF}(x) \le C|x-R_j|^{-6}$ on $A_r \cap \Gamma_j$.
Inserting this and the bound (\ref{ite2}) into the bound from Lemma \ref{inside2}, we obtain
\begin{equation}
\label{kin}
\begin{split}
\tr \left(-\frac{\Delta}{2} \eta_r \gamma^\mathrm{RHF} \eta_r\right)
 &\le  C\left((\lambda r)^{-2}\int_{A_r}\rho^\mathrm{RHF} +\lambda^{-2} r^{-5} + r^{-7}  \right).
\end{split}
\end{equation}
Replacing $r$ by $r/3$ in the above estimate , we get
\begin{equation}
\tr \left(-\frac{\Delta}{2} \eta_{r/3} \gamma^\mathrm{RHF} \eta_{r/3}\right)
\le C\left((\lambda r)^{-2}\int_{A_r}\rho^\mathrm{RHF} +\lambda^{-2} r^{-5} +r^{-7}  \right).
\end{equation}
From Lemma~\ref{outside}, replacing $r$ by $r/3$ and choosing $s=r$, we find that
\begin{align*}
& \int_{A_{r/3}} \rho^\mathrm{RHF}(x) dx \\
&\le
C\sum_{j=1}^K \int_{r/3 \le|x-R_j|<r} \rho^\mathrm{RHF} (x) \, dx + C\left(r^2 \tr (- \Delta \eta_{r/3} \gamma^\mathrm{RHF}\eta_{r/3}) \right)^{3/5} \\
&\quad+  C\left(\sup_{x \in  A_{r/3}} [\phi(x)^{-1}\Phi_{r/3}^\mathrm{RHF}(x)]_+  + r + (\lambda^2 r)^{-1} +  \frac{1}{R_0^2} + \frac{1}{\lambda}\right) .
\end{align*}
Inserting (\ref{ite2})  and(\ref{kin}) into the latter estimate leads to
\begin{align*}
 \int_{A_r} \rho^\mathrm{RHF}(x) dx
 \le \int_{A_{r/3}} \rho^\mathrm{RHF}(x) dx  
&\le C\left(\frac{1}{r^{3}}+\frac{1}{ \lambda^{2}r}\right) \\
&\quad+C\left(\frac{1}{\lambda^{2}} \int_{A_r} \rho^\mathrm{RHF}(x) dx +\frac{1}{\lambda^{2} r^{3}} + \frac{1}{r^{5}}  \right)^{3/5}.
\end{align*}
This proves (\ref{ite3}) immediately. 
Inserting (\ref{ite3}) into (\ref{kin}), we obtain (\ref{ite5}).

Finally, from (\ref{ite5}) and the kinetic Lieb-Thirring inequality, we have
\[
\int_{A_r} (\rho^\mathrm{RHF})^{5/3} \le \int(\eta^2_{r/3}\rho^\mathrm{RHF})^{5/3}
\le C \tr \left(-\frac{\Delta}{2} \eta_{r/3} \gamma^\mathrm{RHF} \eta_{r/3}\right)
\le C\left(\frac{1}{r^7}+ \frac{1}{r^5}\right),
\]
which implies (\ref{ite4}).
\end{proof}
\textit{Step 2}
We introduce the exterior Thomas-Fermi energy functional
\[
\E_{r}^\mathrm{TF}(\rho) = \frac{3}{10}c_\mathrm{TF}\int \rho^{5/3} - \int V_{r} \rho + D[\rho], \quad V_{r}(x) = \chi_r^+\Phi^\mathrm{RHF}_{r}(x).
\]

\begin{lemma}
\label{step2}
The TF functional $\E_{r}^\mathrm{TF}(\rho)$ has a unique minimizer $\rho^\mathrm{TF}_{r}$ over
\[
0\le \rho \in L^{5/3}(\R^3) \cap L^1(\R^3), \quad \int \rho \le Z -  \int_{A_r^c}\rho^\mathrm{RHF}(y) \, dy.
\]
This minimizer is supported on $A_r$ and satisfies the TF equation
\[
c_\mathrm{TF}\rho^\mathrm{TF}_{r}(x)^{2/3} = [\phi_{r}^\mathrm{TF}(x)-\mu_{r}^\mathrm{TF}]_+
\]
with $\phi^\mathrm{TF}_{r}(x) = V_{r}(x) - \rho_{r}^\mathrm{TF} \star |x|^{-1}$ and a constant $\mu_{r}^\mathrm{TF} \ge 0$. 
Moreover, 
\begin{itemize}
\item[(i)]
If $\mu_{r}^\mathrm{TF} > 0$, then 
\[
\int \rho^\mathrm{TF}_{r} = Z -  \int_{A_r^c}\rho^\mathrm{RHF}(y) \, dy.
\]
\item[(ii)]
If (\ref{initialassumption}) holds true for some $\beta$, $D \in (0, 1]$, then
\begin{equation*}
\label{TFminbound}
\int (\rho_{r}^\mathrm{TF})^{5/3} \le Cr^{-7}, \quad \text{for any } r \in (0, D].
\end{equation*}
\end{itemize}
\end{lemma}

\begin{proof}
The existence of $\rho_{r}^\mathrm{TF}$, the TF equation, and (i) follow from~\cite[Theorem 4.1 (i)]{TFDWIC}.
From the TF equation and the fact that $\phi^\mathrm{TF}_{r} \le V_{r} = 0$ on $A_r$, we learn $\mathrm{supp} \, \rho_{r}^\mathrm{TF} \subset A_r$.
Moreover, by the minimizing property of $\rho^\mathrm{TF}_{r}$ and (\ref{ite2}), we obtain
\begin{align*}
0 \ge \E_{r}^\mathrm{TF}(\rho_{r}^\mathrm{TF}) &\ge  \frac{3}{10}c_\mathrm{TF}\int(\rho_{r}^\mathrm{TF})^{5/3} - Cr^{-3}\sum_{j=1}^K \frac{z_j}{Z}\int \frac{\rho_{r}^\mathrm{TF}(x)}{|x-R_j|} \, dx + D[\rho^\mathrm{TF}_{r}] \\
&\ge \frac{3c_\mathrm{TF}}{20} \int (\rho_{r}^\mathrm{TF})^{5/3} -C(r^{-3})^{7/3},
\end{align*}
where we have used $\inf_{\rho \ge 0}\E^\mathrm{TF}(\rho) \ge -C\sum_{j=1}^Kz_j^{7/3}$.
This finishes the proof.
\end{proof}

We will use the next lemma.
\begin{lemma}[Chemical potential estimate]
\label{chemical}
If $\mu_{r}^\mathrm{TF} < \inf_{x \in A_r} \phi_{r}^\mathrm{TF}$, then we have 
$\mu_{r}^\mathrm{TF} = 0$.
\end{lemma}

\begin{proof}
We suppose contrary $\mu_{r}^\mathrm{TF} > 0$.
Then it holds that
\begin{equation}
\label{assumptionLem5}
 \int_{\R^3} \rho_{r}^\mathrm{TF}(y) \, dy = Z -  \int_{A_r^c}\rho^\mathrm{RHF}(y) \, dy.
\end{equation}
By Theorem \ref{Sommerfeld1}, for any $|x-R_j|\ge r$, we have
\[
\nu(\mu_{r}^\mathrm{TF} , r) \le |x-R_j| \phi_{r}^\mathrm{TF}(x).
\]
By definition, we see
\begin{align*}
\nu(\mu_{r}^\mathrm{TF} , r) &\ge \mu_{r}^\mathrm{TF}  \inf_{|x|\ge r} \max \left\{|x|, \frac{c_\mathrm{S}|x|^{-3}}{\mu_{r}^\mathrm{TF} (1+a(r))^2} \right\}\\
&\ge (\mu_r^\mathrm{TF})^{3/4}c_\mathrm{S}^{1/4}(1+a(r))^{-1/2}.
\end{align*}
Moreover, we can estimate that, on some $x \in \Gamma_j$,
\begin{align*}
 \lim_{x \in \Gamma_j, |x-R_j| \to \infty}&|x-R_j| \phi_{r}^\mathrm{TF}(x) \le Z -  \int_{A_r^c}\rho^\mathrm{RHF}(y) \, dy - \int_{\R^3} \rho_{r}^\mathrm{TF}(y) \, dy.
\end{align*}
Hence, we find that
\[
0 < (\mu_r^\mathrm{TF})^{3/4} \le C\left( Z -  \int_{A_r^c}\rho^\mathrm{RHF}(y) \, dy- \int_{\R^3} \rho_{r}^\mathrm{TF}(y) \, dy\right).
\]
Thus, it follows that
\[
 \int_{\R^3} \rho_{r}^\mathrm{TF}(y) \, dy <  Z -  \int_{A_r^c}\rho^\mathrm{RHF}(y) \, dy.
\]
This contradicts the equation (\ref{assumptionLem5}).
\end{proof}
\textit{Step 3}
Now we compare $\rho_{r}^\mathrm{TF}$ with $\1_{A_r}\rho^\mathrm{TF}$.
\begin{lemma}
\label{step3}
Let $\tilde r = R_0^{-1}r^{\frac{\xi}{\xi + \eta}}R_\mathrm{min}^{\frac{\eta}{\xi+\eta}}$.
We can choose a universal constant $\beta > 0$ small enough such that, if (\ref{initialassumption}) holds for some $D \in [Z^{-1/3}, R_0]$ and if $r \in [Z^{-1/3}, D]$, then $\mu_{r}^\mathrm{TF} = 0$ and for any $s \in [r, \tilde r]$
\begin{align}
\label{step3a}
\sup_{x\in \partial A_s}|\phi_{r}^\mathrm{TF}(x) - \phi^\mathrm{TF}(x)|&\le C(r/s)^\xi s^{-4}, \\
\label{step3b}
\sup_{x\in \partial A_s}|\rho_{r}^\mathrm{TF}(x) - \rho^\mathrm{TF}(x)| &\le C(r/s)^\xi s^{-6}.
\end{align}
 Here $\xi = (\sqrt 73 - 7)/2 \sim 0.77$.
\end{lemma}
\begin{proof}
We recall Theorem \ref{Sommerfeld1}, that is, in $A_r \cap \Gamma_j$
\begin{equation}
\begin{split}
\label{som1}
K\left(1 + A(r)\left(\frac{r}{|x-R_j|} \right)^\xi \right)\ge \frac{\phi^\mathrm{TF}(x)}{c_\mathrm{s}|x-R_j|^{-4}} 
\ge \left(1+a(r)\left(\frac{r}{|x-R_j|} \right)^\xi \right)^{-2},
\end{split}
\end{equation}
\begin{align*}
K^{3/2}\left(1 + A(r)\left(\frac{r}{|x-R_j|} \right)^\xi \right)^{3/2} \ge \frac{\rho^\mathrm{TF}(x)}{\left(\frac{c_\mathrm{s}}{c_\mathrm{TF}} \right)^{3/2}|x-R_j|^{-6}}
\ge \left(1+a(r)\left(\frac{r}{|x-R_j|} \right)^\xi \right)^{-3}.
\end{align*}
From this, we have $C|x-R_j|^{-6} \ge \rho^\mathrm{TF}(x) \ge C^{-1}|x-R_j|^{-6}$ for $x \in A_r \cap \Gamma_j$, and hence
\begin{equation}
\label{densitybd}
Cr^{-3} \ge \int_{A_r } \rho^\mathrm{TF} (x) \ge C^{-1} r^{-3}
\end{equation}
for any $r \in [Z^{-1/3}, R_0]$.

\begin{lemma}
For every $r \in (0, R_0]$, we have
\[
\widetilde \E_{r}(\chi_r^+\rho^\mathrm{TF}) \le \widetilde \E_{r}(\rho)
\]
for all $0 \le \rho \in L^{5/3}(\R^3)\cap L^1(\R^3)$ with $\mathrm{supp} \, \rho \subset A_r$, where
\[
\widetilde \E_{r}(\rho) =  \frac{3}{10}c_\mathrm{TF} \int \rho^{5/3} - \int \Phi^\mathrm{TF}_{r}\rho + D[\rho].
\]
\end{lemma}

\begin{proof}
For all $0 \le \rho \in L^{5/3}(\R^3)\cap L^1(\R^3)$ with $\mathrm{supp} \, \rho \subset A_r$, by the minimality of $\rho^\mathrm{TF}$ we have
\[
\E^\mathrm{TF}(\rho^\mathrm{TF}) \le \E^\mathrm{TF}(\1_{A_r^c}\rho^\mathrm{TF} + \rho).
\]
Since $\1_{A_r^c}\rho^\mathrm{TF}$ and $\rho$ have disjoint supports, we can write
\begin{align*}
\E^\mathrm{TF}(\1_{A_r^c}\rho^\mathrm{TF} + \rho)
&= \E^\mathrm{TF}(\1_{A_r^c}\rho^\mathrm{TF}) + \E^\mathrm{TF}(\rho)
+\iint_{A_r^c} \frac{\rho(x) \rho^\mathrm{TF}(y)}{|x-y|} \, dx \, dy\\
&=\E^\mathrm{TF}(\1_{A_r^c}\rho^\mathrm{TF}) + \widetilde \E_{r}(\rho).
\end{align*}
In particular, we can apply the latter equality with $\rho=\chi_r^+\rho^\mathrm{TF}$ and obtain
\begin{align*}
\E^\mathrm{TF}(\rho^\mathrm{TF}) &= \E^\mathrm{TF}(\1_{A_r^c}\rho^\mathrm{TF}+\chi_r^+\rho^\mathrm{TF})\\
&=\E^\mathrm{TF}(\1_{A_r^c}\rho^\mathrm{TF}) + \widetilde \E_{r}(\chi_r^+\rho^\mathrm{TF}).
\end{align*}
Thus
\[
0 \le \E^\mathrm{TF}(\1_{A_r^c}\rho^\mathrm{TF} + \rho) -\E^\mathrm{TF}(\rho^\mathrm{TF}) = \widetilde \E_{r}(\rho) - \widetilde \E_{r}(\chi_r^+\rho^\mathrm{TF}).
\]
This completes the proof.
\end{proof}
Now using this lemma with $\rho = \rho_{r}^\mathrm{TF}$ and the identity
\[
\widetilde \E_{r}(\rho) = \E_{r}^\mathrm{TF}(\rho) + \int (\Phi_{r}^\mathrm{RHF} - \Phi_{r}^\mathrm{TF})\rho,
\]
we find that
\begin{equation}
\label{6red}
\E^\mathrm{TF}_{r}(\chi_r^+\rho^\mathrm{TF}) \le \E_{r}^\mathrm{TF}(\rho_{r}^\mathrm{TF}) - \int (\Phi_{r}^\mathrm{RHF} - \Phi_{r}^\mathrm{TF})(\chi_r^+\rho^\mathrm{TF} - \rho_{r}^\mathrm{TF}). 
\end{equation}
Since $\Phi_{r}^\mathrm{RHF}(x) - \Phi_{r}^\mathrm{TF}(x)$ is harmonic in $A_r$, we deduce from (\ref{initialassumption}) that
\[
\sup_{x \in A_r}|\Phi_{r}^\mathrm{RHF}(x) - \Phi_{r}^\mathrm{TF}(x)| = \sup_{x \in \partial A_r}|\Phi_{r}^\mathrm{RHF}(x) - \Phi_{r}^\mathrm{TF}(x)| \le \beta r^{-4}.
\]
Therefore, we get
\begin{align*}
\left|\int (\Phi_{r}^\mathrm{RHF} - \Phi_{r}^\mathrm{TF})(\chi_r^+\rho^\mathrm{TF} - \rho^\mathrm{TF}) \right| &\le \beta r^{-4} \int(\chi_r^+\rho^\mathrm{TF} + \rho_{r}^\mathrm{TF}) \\
&\le C \beta r^{-7},
\end{align*}
where we have used the upper bound in (\ref{densitybd}). Moreover, by (\ref{ite2}) and the assumption $N \ge Z$, we see
\[
\int \rho_{r}^\mathrm{TF} \le Z - \int_{A_r^c} \rho^\mathrm{RHF}(x) \, dx \le \int_{A_r} \rho^\mathrm{RHF} \le Cr^{-3}.
\]
Hence (\ref{6red}) reduces to
\begin{equation}
\label{6ded}
\E^\mathrm{TF}_{r}(\chi_r^+\rho^\mathrm{TF}) \le \E_{r}^\mathrm{TF}(\rho_r^\mathrm{TF}) +C \beta r^{-7}.
\end{equation}
We want to compare $\chi_r^+\rho^\mathrm{TF}$ with $\rho_{r}^\mathrm{TF}$ using the minimality property of the latter as in~\cite[Proof of Lemma 6.8]{TFDWIC}.
Using  (\ref{ite1}), (\ref{densitybd}), we have
\begin{align*}
\int_{A_r} \rho^\mathrm{TF}(x) \, dx -  \left( Z - \int_{A_r^c} \rho^\mathrm{RHF}(y) \, dy \right)
 &\le \int_{A_r^c} (\rho^\mathrm{RHF} - \rho^\mathrm{TF} )
 \le C\beta \int_{A_r} \rho^\mathrm{TF}. 
 \end{align*}
 This can be rewritten as
 \begin{equation}
 \label{rewr}
 \int_{A_r}(1-C\beta) \rho^\mathrm{TF} \le  \left( Z - \int_{A_r^c} \rho^\mathrm{RHF}(y) \, dy \right).
 \end{equation}
 In the following, we choose $\beta > 0$ small enough so that $C \beta \le 1/2$.
 Since $\int (C\rho)^{5/3} +D[C\rho] \le \int \rho^{5/3} +D[\rho]$ for $C \le 1$, using (\ref{ite2}) and (\ref{densitybd}), we may estimate
 \[
 \E_{r}^\mathrm{TF}((1-C \beta)\chi_r^+\rho^\mathrm{TF}) - \E_r^\mathrm{TF}(\chi_r^+ \rho^\mathrm{TF}) \le  C\beta \int_{A_r} \Phi^\mathrm{RHF}_{r} \rho^\mathrm{TF}
 \le C \beta r^{-7}.
 \]
 Therefore, from (\ref{6ded}) we derive that
 \[
  \E_{r}^\mathrm{TF}((1-C \beta)\chi_r^+\rho^\mathrm{TF}) \le  \E_{r}^\mathrm{TF}(\rho_{r}^\mathrm{TF}) + C \beta r^{-7}.
 \]
 Combining with (\ref{rewr}) and the minimality of $\rho_{r}^\mathrm{TF}$, we obtain
 \begin{align*}
   \E_{r}^\mathrm{TF}((1-C \beta)\chi_r^+ \rho^\mathrm{TF}) +   \E_{r}^\mathrm{TF}(\rho_{r}^\mathrm{TF}) -2\E_{r}^\mathrm{TF}\left(\frac{(1- C\beta)\chi_r^+ \rho^\mathrm{TF} + \rho_{r}^\mathrm{TF}}{2}\right)
    \le C\beta r^{-7}.
 \end{align*}
 By the convexity of $\rho^{5/3}$ and $D[\rho]$, we have
 \begin{equation}
 \label{6direct}
  D[(1-C \beta)\chi_r^+\rho^\mathrm{TF} - \rho_{r}^\mathrm{TF}] \le C\beta r^{-7}.
 \end{equation}
 We also derive that
 \begin{equation}
 \label{6rec}
 \begin{split}
 \int &\bigg[ \left( (1-C \beta)\chi_r^+\rho^\mathrm{TF}(x)\right)^{5/3} +\rho_{r}^\mathrm{TF}(x)^{5/3} \\
 &-2\left(\frac{(1-C \beta)\chi_r^+\rho^\mathrm{TF}(x)) + \rho_{r}^\mathrm{TF}(x)}{2} \right)^{5/3} \bigg] \, dx \le C \beta r^{-7}.
 \end{split}
 \end{equation}
 From (\ref{6direct}) and the convexity of Coulomb term $D[\cdot]$, we learn that
  \begin{equation}
  \label{6direct2}
 \begin{split}
  D[\chi_r^+ \rho^\mathrm{TF} - \rho_{r}^\mathrm{TF}]
  &\le 2D[\chi_r^+\rho^\mathrm{TF} - (1-C \beta)\chi_r^+ \rho^\mathrm{TF}]  +2D[(1-C \beta )\chi_r^+\rho^\mathrm{TF} -\rho_{r}^\mathrm{TF}] \\
     &\le ( C \beta)^2D[\chi_r^+ \rho^\mathrm{TF}] +C \beta r^{-7}\\
   &\le C\beta r^{-7},
 \end{split}
 \end{equation}
  where the last inequality follows from choosing $C\beta \le 1$.
  
  Now we apply the fact that $f\star |x|^{-1} \le C\|f\|^{5/7}_{L^{5/3}}D[f]^{1/7}$ (see~\cite[Eq. (6.3)]{TFDWIC}) with $f = \pm (\chi_r^+\rho^\mathrm{TF} - \rho_{r}^\mathrm{TF})$. 
  Then, using  (\ref{TFminbound}) and $\int_{A_r} (\rho^\mathrm{TF})^{5/3} \le Cr^{-7}$,we have
  \begin{align*}
  |(\chi_r^+ \rho^\mathrm{TF} - \rho_{r}^\mathrm{TF}) \star|x|^{-1}|
  \le C \beta^{1/7} r^{-4}.
  \end{align*}
  Combining this with the assumption (\ref{initialassumption}), we get
   \begin{align*}
  |\phi_{r}^\mathrm{TF}(x) - \phi^\mathrm{TF}(x)|
  &=
  |\Phi_{r}^\mathrm{RHF}(x) - \Phi_{r}^\mathrm{TF}(x) + (\chi_r^+\rho^\mathrm{TF} - \rho_{r}^\mathrm{TF})\star |x|^{-1}| \\
 &\le  C(\beta + \beta^{1/7}) r^{-4}, \quad \text{for any } x \in A_r.
  \end{align*}
  
We note that $Cr^{-4} \ge \phi^\mathrm{TF}(x) \ge C^{-1}r^{-4}$ for $x \in A_r$ by the Sommerfeld bound. Therefore, if $\beta > 0$ is sufficiently small, we see
\begin{equation}
\label{6TF1}
Cr^{-4} \ge \phi_{r}^\mathrm{TF}(x) \ge C^{-1}r^{-4}, \quad \text{for all } x \in A_r.
\end{equation}
To improve this bound, we need to show that $\mu_{r}^\mathrm{TF} = 0$.
This follows from Lemma~\ref{chemical} if 
\begin{equation}
\label{assumption}
\mu_{r}^\mathrm{TF} < \inf_{ x \in \partial A_r } \phi_{r}^\mathrm{TF}(x).
\end{equation}
We now suppose that (\ref{assumption}) fails. Then from (\ref{6TF1}) we find that
\[
\mu_{r}^\mathrm{TF} \ge \inf_{ x \in A_r}\phi_{r}^\mathrm{TF}(x) \ge C^{-1}r^{-4}.
\]
On the other hand, $\phi_{r}^\mathrm{TF}(x) \le \Phi_{r}^\mathrm{RHF}(x) \le Cr^{-3}\phi(x)$ by (\ref{ite2}).
Therefore, from the TF equation
\[
c_\mathrm{TF} \rho_{r}^\mathrm{TF}(x)^{2/3} =[\phi_{r}^\mathrm{TF}(x) - \mu_{r}^\mathrm{TF}]_+ \le  \left[Cr^{-3}\phi(x) - C^{-1}r^{-4}\right]_+,
\]
we see $\rho_{r}^\mathrm{TF}(x) = 0$ on $A_{C^2r}$.
Since the integrand in  (\ref{6rec}) is pointwise nonnegative, we can restrict the integral on $A_{C^2r}$. Then, using $\rho_{r}^\mathrm{TF}(x) = 0$ on $A_{C^2r}$, we derive from (\ref{6rec}) that
\begin{align*}
C\beta r^{-7} \ge \int_{A_{C^2r}} \left((1-C \beta)\ \rho^\mathrm{TF}(x)\right)^{5/3} \, dx
\ge  C^{-1} (1-C \beta)^{5/3}r^{-7}.
\end{align*}
Thus we get $C^{-1} (1-C \beta)^{5/3}r^{-7} \le C\beta r^{-7}$ and a contradiction if $\beta > 0$ is sufficiently small. Then we can choose $\beta >0$ small enough such that $\mu_{r}^\mathrm{TF} = 0$.
Hence we can use Theorem~\ref{Sommerfeld1} and  Theorem~\ref{Sommerfeld2} for $\phi^\mathrm{TF}$ and $\phi_r^\mathrm{TF}$, and therefore we arrive at, for $x \in A_r\cap \Gamma_j$,
\[
|\phi_r^\mathrm{TF}(x) - \phi^\mathrm{TF}(x)|
\le c_s|x-R_j|^{-4}\left(A_1^j(r)\left(\frac{|x-R_j|}{R_j}\right)^\eta + (A_2^j(r) + 2a(r)\left(\frac{r}{|x-R_j|}\right)^\xi\right),
\]
where we have used the fact that $(1+t)^{-2} \ge 1-2t$ for $t \in (-1, \infty)$.
Since $s \le \tilde r$ it holds that $(s/R_j)^\eta \le 2^\eta R_0^{-(\xi + \eta)}(r/s)^\xi$.
If we note that $A_i^j(r) \le C$ and $a(r) \le C$ by (\ref{6TF1}), then (\ref{step3a}) follows.
Proceeding this way, one can arrive at (\ref{step3b}) from the fact that, for any $t\in(0, T]$, $(1+t)^{3/2} \le 1 + t((1+T)^{3/2} - 1)T^{-1}$.
Then the proof is complete.
\end{proof}
\textit{Step 4}
In this step, we compare $\rho_{r}^\mathrm{TF}$ with $\1_{A_r} \rho^\mathrm{RHF}$.
\begin{lemma}
\label{step4}
Let $\beta > 0$ be as in Lemma \ref{step3}. We assume that (\ref{initialassumption}) holds for some $D \in [Z^{-1/3}, R_0]$. Then, if $r \in [Z^{-1/3}, D]$, we have
\[
D[\rho_{r}^\mathrm{TF} - \1_{A_r}\rho^\mathrm{RHF}] \le Cr^{-7+1/3}.
\]
\end{lemma}
\begin{proof}
\underline{\textit{Upper Bound.}}
We will prove that
\begin{equation}
\label{6upper}
\E_r^\mathrm{RHF}(\eta_r\gamma^\mathrm{RHF} \eta_r) \le \E_{r}^\mathrm{TF}(\rho_{r}^\mathrm{TF}) + Cr^{-7}(r^{2/3} + \lambda^{-2} r^2 + \lambda).
\end{equation}
We use Lemma \ref{semicl} (2) with $V_{r}' \coloneqq \1_{A_{r+s}} \phi_{r}^\mathrm{TF}$, $s \le r$ to be chosen later, and $g$ spherically symmetric to obtain a density matrix $\widetilde \gamma$ as in the statement. Since $\mu_{r}^\mathrm{TF} = 0$ by Lemma \ref{step3}, we deduce from the TF equation in Lemma \ref{step2} that
\[
\rho_{\widetilde \gamma} =2^{5/2}(6\pi^2)^{-1}\left(\1_{A_{r+s}}(\phi_{r}^\mathrm{TF})^{3/2} \right)\star g^2 = (\1_{A_{r+s}} \rho_{r}^\mathrm{TF})\star g^2.
\]
Since $\rho_{\widetilde \gamma}$ is supported in $A_r$ and
\[
\tr \widetilde \gamma = \int \rho_{\widetilde \gamma} = \int_{A_{r+s}}  \rho_{r}^\mathrm{TF} 
\le
\int \rho_{r}^\mathrm{TF}  \le \int_{A_{r}}  \rho^\mathrm{RHF},
\]
we may apply Lemma~\ref{frominside} and obtain
$\E_r^\mathrm{RHF}(\eta_r\gamma^\mathrm{RHF}\eta_r)
\le \E_r^\mathrm{RHF}(\widetilde \gamma) + \mathcal{R}$.
Next, we bound $\E^\mathrm{RHF}(\widetilde \gamma)$.
By the semiclassical estimate from Lemma~\ref{semicl} (2), we have
\begin{equation}
\begin{split}
\E^\mathrm{RHF}(\widetilde \gamma)
&\le2^{3/2}(5\pi^2)^{-1} \int [V_r]^{5/2}_+ + Cs^{-2} \int [V_{r}]^{3/2}_+
+ D[\rho_{r}^\mathrm{TF}\star g^2]  -\int \Phi^\mathrm{RHF}_{r} (\1_{A_{r+s}}  \rho_{r}^\mathrm{TF})\star g^2 \\
&\le
 2^{3/2}(5\pi^2)^{-1} \int [\phi_{r}^\mathrm{TF}]^{5/2}_+ + Cs^{-2} \int \rho_{r}^\mathrm{TF} - \int_{A_r} \Phi^\mathrm{RHF}_{r} \rho_{r}^\mathrm{TF} \\
&\quad+ D[\rho_r^\mathrm{TF}]
+ \int_{A_{r+s} } (\Phi^\mathrm{RHF}_{r} -\Phi^\mathrm{RHF}_{r} \star g^2)\rho_{r}^\mathrm{TF}
+\int_{A_r \cap A_{r+s}^c} \Phi_{r}^\mathrm{RHF} \rho_r^\mathrm{TF} \\
&\le \E_r^\mathrm{TF}(\rho_r^\mathrm{TF}) +Cs^{-2} \int \rho_{r}^\mathrm{TF} +\int_{A_r \cap A_{r+s}^c} \Phi^\mathrm{RHF}_{r} \rho_{r}^\mathrm{TF},
\end{split}
\end{equation}
where we have used $\Phi^\mathrm{RHF}_{r} \star g^2 \ge \Phi^\mathrm{RHF}_{r}$ on $A_r$ in the second inequality.
This fact follows from Newton's theorem and the assumption $s \le r$.
According to (\ref{ite3}), we get
\[
 \int \rho_{r}^\mathrm{TF} \le  \int_{A_r} \rho^\mathrm{RHF} \le Cr^{-3}.
\]
We note that $\rho_{r}^\mathrm{TF}(x) \le C|x-R_j|^{-6}$ on $A_r \cap \Gamma_j$ and $x \in \Gamma_j$ if $r \le |x-R_j| < r+s$. Then
\[
\int_{A_r \cap A_{r+s}^c} \Phi^\mathrm{RHF}_{r} \rho_{r}^\mathrm{TF} 
\le Cr^{-3} \sum_{j=1}^K\int_{r \le |x-R_j| \le r+s} |x-R_j|^{-7}  \, dx \le Csr^{-8}.
\]
We choose $s=r^{5/3}$ and get
\[
\E_r^\mathrm{RHF}(\widetilde \gamma) \le
\E_{r}^\mathrm{TF}(\rho_{r}^\mathrm{TF}) + Cr^{-7+2/3}.
\]
Finally, since $\lambda \le 1/2$, we have
\begin{align*}
\mathcal{R} \le C(\lambda^{-2}r^{-5} + \lambda r^{-7}).
\end{align*}
Hence we obtain the desired upper bound.

\underline{\textit{Lower bound}}
We will prove
\[
\E_r^\mathrm{RHF}(\eta_r \gamma^\mathrm{RHF} \eta_r) \ge \E_{r}^\mathrm{TF}(\rho_{r}^\mathrm{TF}) + D[\eta_r^2 \rho^\mathrm{RHF} - \rho_{r}^\mathrm{TF}] -Cr^{-7 +1/3}.
\]
We can estimate
\begin{align*}
\E_r^\mathrm{RHF}(\eta_r \gamma^\mathrm{RHF} \eta_r) 
&= \tr \left[ \left(-\frac{\Delta}{2} - \phi_{r}^\mathrm{TF}\right) \eta_r \gamma^\mathrm{RHF} \eta_r \right]
+D[\eta_r^2\rho^\mathrm{RHF} - \rho_{r}^\mathrm{TF}] + D[ \rho_{r}^\mathrm{TF}]\\
&\ge -2^{5/2}(15\pi^2)^{-1}\int[\phi_{r}^\mathrm{TF}]_+^{5/2} -Cs^{-2}\int \eta_r^2 \rho^\mathrm{RHF} \\
&\quad -C\left(\int[\phi_{r}^\mathrm{TF}]_+^{5/2} \right)^{3/5} \left( \int[\phi_{r}^\mathrm{TF}- \phi_{r}^\mathrm{TF}\star g^2]_+^{5/2}\right)^{2/5} \\
&\quad + D[\eta_r^2\rho^\mathrm{RHF} - \rho_{r}^\mathrm{TF}]  -D[\rho_{r}^\mathrm{TF}] \\
&=  \E_{r}^\mathrm{TF}(\rho_{r}^\mathrm{TF}) + D[\eta_r^2 \rho^\mathrm{RHF} - \rho_{r}^\mathrm{TF}] -Cs^{-2}\int \eta_r^2 \rho^\mathrm{RHF} \\
&\quad -C\left(\int[\phi_{r}^\mathrm{TF}]_+^{5/2} \right)^{3/5} \left( \int[\phi_{r}^\mathrm{TF}- \phi_{r}^\mathrm{TF}\star g^2]_+^{5/2}\right)^{2/5}.
\end{align*}
We note that
\begin{align*}
\int \eta_r^2 \rho^\mathrm{RHF}  &\le Cr^{-3}, \\
\int[\phi_{r}^\mathrm{TF}]_+^{5/2} =C\int (\rho_{r}^\mathrm{TF})^{5/3} &\le Cr^{-7}.
\end{align*}
We know $|x|^{-1} - |x|^{-1}\star g^2 \ge 0 $ and thus $\rho_{r}^\mathrm{TF} \star(|x|^{-1} - |x|^{-1}\star g^2) \ge0$.
Since the TF equation $\phi_{r}^\mathrm{TF} =\chi_r^+\Phi_{r}^\mathrm{RHF} -\rho_{r}^\mathrm{TF}\star |x|^{-1}$, we have 
\[
\phi_{r}^\mathrm{TF} - \phi_{r}^\mathrm{TF} \star g^2 \le \chi_r^+\Phi_{r}^\mathrm{RHF}  -(\chi_r^+\Phi_{r}^\mathrm{RHF} )\star g^2 \eqqcolon f.
\]
By Newton's theorem, we infer that $\mathrm{supp} f \subset \bigcup_{j=1}^K\{x\colon r-s \le |x-R_j| \le r+s\}$.
Hence, by $|f(x)| \le Cr^{-4}$, we have
\[
[\phi_{r}^\mathrm{TF} - \phi_{r}^\mathrm{TF} \star g^2]_+ \le Cr^{-4}\sum_{j=1}^K\1(r-s \le |x-R_j| \le r+s).
\]
Together with these facts, we learn
\[
\int[\phi_{r}^\mathrm{TF} - \phi_{r}^\mathrm{TF} \star g^2]_+^{5/2}\le Cr^{-8}s.
\]
We conclude that
\begin{align*}
\E_r^\mathrm{RHF}(\eta_r \gamma^\mathrm{RHF} \eta_r)  &\ge 
 \E_{r}^\mathrm{TF}(\rho_{r}^\mathrm{TF}) + D[\eta_r^2 \rho^\mathrm{RHF} - \rho_{r}^\mathrm{TF}]  -C(s^{-2}r^{-3} + r^{-37/5}s^{2/5}).
\end{align*}
Then  we choose $s = r^{11/6}$ and arrive at the desired lower bound.

\underline{\it Conclusion}
Combining the upper and lower bound, we learn
\[
D[\eta_r^2 \rho^\mathrm{RHF} - \rho_{r}^\mathrm{TF}] \le Cr^{-7}(r^{1/3} + \lambda^{-2}r^2 +\lambda).
\]
Using the Hardy-Littlewood-Sobolev inequality, we have
\begin{align*}
D[\chi_r^+\rho^\mathrm{RHF} - \eta_r^2 \rho^\mathrm{RHF}]
&\le C\|\1_{A_r \cap A_{(1+\lambda)r}^c}\rho^\mathrm{RHF}\|^2_{L^{6/5}} \\
&\le C\left(\int_{A_r}\rho^\mathrm{RHF}(x)^{5/3} \, dx \right)^{6/5}\left(\sum_{j=1}^K  \int_{r \le |x-R_j| \le (1+\lambda)r} \, dx\right)^{7/15} \\
&=C\lambda^{7/15} r^{-7}.
\end{align*}
By convexity, we see
\begin{align*}
D[\chi_r^+\rho^\mathrm{RHF} - \rho_{r}^\mathrm{TF}] &\le 2D[\chi_r^+\rho^\mathrm{RHF} - \eta_r^2 \rho^\mathrm{RHF}]
+2D[\eta_r^2\rho^\mathrm{RHF} - \rho_{r}^\mathrm{TF}] \\
&\le Cr^{-7}(\lambda^{7/15} + r^{1/3} + \lambda^{-2}r^2),
\end{align*}
for any $\lambda \in (0, 1/2]$.
We choose $\lambda = r^{30/37}$ and get
\[
D[\chi_r^+\rho^\mathrm{RHF} - \rho_{r}^\mathrm{TF}] \le Cr^{-7+1/3}.
\]
This completes the proof.

\textit{Step 5}

We turn to prove Theorem~\ref{Thm.iterative}.
Let $r=D^{1+\delta}$,  $s\in [r^{-(1+\delta)},  \min\{r^{(1-\delta)/(1+\delta)}, \tilde r\}]$, and $x \in \partial A_s$.
Now we choose a constant $\delta \in (0, 1)$ such that
\begin{align*}
\frac{1+\delta}{1-\delta}\left(\frac{49}{36}-a \right) &< \frac{49}{36} \\
\frac{1}{36} - \frac{10\delta}{1-\delta} &>0.
\end{align*}
We consider two cases.

{\it Case 1} $D^{1+\delta} \le Z^{-1/3}$.

By the initial step, for any $s \le r^{(1-\delta)/(1+\delta)} \le (Z^{-1/3})^{(1-\delta)/(1+\delta)}$, we have
\begin{align*}
|\Phi^\mathrm{RHF}_{s}(x) - \Phi^\mathrm{TF}_{s}(x)| &\le CZ^{49/36 - a}s^{1/12}
\\
&\le Cs^{1/12 - 3\frac{1+\delta}{1-\delta}(49/36 - a)}  \\
&=Cs^{-4+\epsilon_1},
\end{align*}
which is the desired conclusion.

{\it Case 2} $D^{1+\delta} \ge Z^{-1/3}$.

We may split
\begin{align*}
\Phi^\mathrm{RHF}_{s}(x) - \Phi^\mathrm{TF}_{s}(x) &= \phi_{r}^\mathrm{TF}(x) - \phi^\mathrm{TF}(x) +\int_{A_s}\frac{\rho_{r}^\mathrm{TF}(y) - \rho^\mathrm{TF}(y)}{|x-y|} \, dy\\
&\quad + \sum_{i=1}^K\int_{|y-R_i| < s}\frac{\rho_{r}^\mathrm{TF}(y) - \chi_r^+\rho^\mathrm{RHF}(y)}{|x-y|} \, dy.
\end{align*}
We know from Lemma~\ref{step3} that
\[
| \phi_{r}^\mathrm{TF}(x) - \phi^\mathrm{TF}(x)| \le C\left(\frac{r}{s} \right)^\xi s^{-4}
\]
and
\[
\int_{A_r}\frac{\rho_{r}^\mathrm{TF}(y) - \rho^\mathrm{TF}(y)}{|x-y|} \, dy
\le C\left(\frac{r}{s} \right)^\xi s^{-4}.
\]

We note that $\1_{(|y-R_i| < s)}(\rho_r^\mathrm{TF} - \chi_r^+ \rho^\mathrm{RHF})\star |x|^{-1}$ is harmonic in $|x-R_i| \ge s$ for any $i = 1, \dots, K$.
Hence we get from Lemma~\ref{coulomb} that
\begin{align*}
\left| \int_{|y-R_i| < s}\frac{\rho_{r}^\mathrm{TF}(y) - \chi_r^+\rho^\mathrm{RHF}(y)}{|x-y|} \, dy\right|
&\le \sup_{|x-R_i| = s}\left| \int_{|y-R_i| < s}\frac{\rho_{r}^\mathrm{TF}(y) - \chi_r^+\rho^\mathrm{RHF}(y)}{|x-y|}\, dy \right| \\
&\le C\|\rho_{r}^\mathrm{TF} - \chi_r^+\rho^\mathrm{RHF}\|^\mathrm{5/6}_{L^{5/3}}\left(sD[\rho_{r}^\mathrm{TF} - \chi_r^+\rho^\mathrm{RHF}] \right)^{1/12} \\
&\le Cs^{-7/2}(r^{-7+1/3}s)^{1/12} \\
&=Cs^{-4+1/36}\left(\frac{s}{r} \right)^{4+1/12 -1/36}.
\end{align*}
In conclusion,
\begin{equation}
\label{step5}
\sup_{x \in \partial A_s}|\Phi^\mathrm{RHF}_{s}(x) - \Phi^\mathrm{TF}_{s}(x)|
\le C\left( \frac{r}{s} \right)^\xi s^{-4} + C\left( \frac{s}{r}\right)^5 s^{-4 + 1/36}.
\end{equation}

 For any $D \le s \le D^{1-\delta}$ we learn
\[
s^{2\delta/(1-\delta)} \le r/s \le s^{\delta}.
\]
Thus we deduce from  (\ref{step5})  that
\[
|\Phi^\mathrm{RHF}_{s}(x) - \Phi^\mathrm{TF}_{s}(x)| \le Cs^{-4 + \xi\delta} + Cs^{-4 + 1/36 - 10\delta/(1-\delta)}\le Cs^{-4 + \epsilon_2} .
\]
Then the proof is complete.
\end{proof}

\section{Screened potential estimate}
Now we can prove the following theorem.
\begin{theorem}[screened potential estimate]
\label{screened potential}
There are universal constants $C, \epsilon, D > 0$ such that
\[
\sup_{x \in \partial A_r} \left|\Phi^\mathrm{RHF}_{r}(x) - \Phi_{r}^\mathrm{TF}(x) \right| \le Cr^{-4+\epsilon} \quad \text{for any } r \le D.
\]
\end{theorem}

\begin{proof}
The proof is essentially the same as that of \cite[Theorem 5.1]{Samojlow}.
Let $\sigma = \max\{C_1, C_2\}$. We may assume $\beta < \sigma$.
We put $D_0 = Z^{-1/3}$.
From Lemma~\ref{initial} we learn
\begin{equation}
\label{eq.initial}
\sup_{x \in \partial A_r} \left|\Phi^\mathrm{RHF}_{r}(x) - \Phi_{r}^\mathrm{TF}(x)\right| \le \sigma r^{-4+\epsilon} \quad \text{for any } r \le D_0 = Z^{-1/3}.
\end{equation}
Now we define
\[
M \coloneqq \sup\left\{ r \in \R \colon \sup_{x \in \partial A_s} \left|\Phi^\mathrm{RHF}_{s}(x) - \Phi_{s}^\mathrm{TF}(x) \right| \le \sigma s^{-4+\epsilon}, \text{ for any } s \le r^{\frac{1}{1+\delta}} \right\}.
\]
Next, we suppose that
\begin{enumerate}
\item
$M < R_0$

and

\item
$(M^{\frac{1}{1+\delta}}, \min\{ M^{\frac{1-\delta}{1+\delta}}, \tilde M\}) \neq \emptyset$,
\end{enumerate}
where $\tilde M \coloneqq R_0^{-1}M^{\xi/(\xi + \eta)} R_\mathrm{min}^{\eta/(\xi + \eta)}$.
If $D_0 <M$, then there is a sequence  such that $D_{n} \to M$ and $D_{0}  \le D_n \le M$ for large $n$.
From this and Theorem~\ref{Thm.iterative}, we see
\[
\sup_{x \in \partial A_r} \left|\Phi^\mathrm{RHF}_{r}(x) - \Phi_{r}^\mathrm{TF}(x)\right| \le \sigma r^{-4+\epsilon}, \quad \text{for any } r \in \left[D_{n}^{\frac{1}{1+\delta}}, \min \left\{D_{n}^{\frac{1-\delta}{1+\delta}}, \tilde D_{n}\right\}\right],
\]
where $\tilde D_n \coloneqq R_0^{-1}D_n^{\xi/(\xi + \eta)} R_\mathrm{min}^{\eta/(\xi + \eta)}$.
From (2), we have
\[
M^{\frac{1}{1+\delta}} \in \left(D_{n}^{\frac{1}{1+\delta}}, \min \left\{D_{n}^{\frac{1-\delta}{1+\delta}}, \tilde D_{n}\right\}\right) \neq \emptyset
\]
for large $n$. This contradicts the definition of $M$.
If $D_0=M$, then $D_0 \le R_0$ and
\[
\sup_{x \in \partial A_r} \left|\Phi^\mathrm{RHF}_{r}(x) - \Phi_{r}^\mathrm{TF}(x) \right| \le \sigma r^{-4+\epsilon}, \quad \text{for any } r \le \min\{M^{\frac{1-\delta}{1+\delta}}, \tilde M\},
\]
which also contradicts the definition of $M$.
Finally, if $D_0 > M$ then we can choose $M' \in (M, \min\{1, D_0\})$.
Then (\ref{eq.initial}) leads to a contradiction.
Hence at least one of (1) and (2) cannot hold.
If (1) is true, then $M \ge cR_\mathrm{min}^{\frac{\eta(1+\delta)}{\eta-\delta\xi}}$.
Therefore we arrive at
\[
M \ge \min\left\{R_0, cR_\mathrm{min}^{\frac{\eta(1+\delta)}{\eta-\delta\xi}}\right\} \ge D^{1+\delta},
\]
where $D$ is the desired universal constant.
Then the theorem follows.
\end{proof}

\section*{Proof of Theorem \ref{theorem.main}}
Since $N \le 2Z + K$~\cite{Lieb1984}, we need only consider the case $N \ge Z \ge 1$.
By Theorem~\ref{screened potential}, there are universal constants $C, \epsilon, D >0$ such that
\[
\sup_{x \in \partial A_r} \left|\Phi^\mathrm{RHF}_{r}(x) - \Phi_{r}^\mathrm{TF}(x) \right|  \le Cr^{-4+\epsilon}, \quad
\text{for any }  r \le D.
\]
Hence we can use (\ref{initialassumption}) with a universal constant $\beta = CD^{\epsilon}$.
Now we choose $D$ sufficiently small so that $D \le 1$ and $\beta \le 1$.
By applying Lemma~\ref{lem.ite} and using (\ref{ite1}) and (\ref{ite3}) with $r = D$, we infer that
\[
\int_{A_D} \rho^\mathrm{RHF} + \left|\sum_{j=1}^K \int_{|x-R_j|<D} (\rho^\mathrm{RHF} - \rho^\mathrm{TF}) \right| \le C.
\]
By $\int \rho^\mathrm{TF} = Z$, we have
\[
N=\int\rho^\mathrm{RHF} = \int_{A_D} \rho^\mathrm{RHF} +\sum_{j=1}^K \int_{|x-R_j|<D} (\rho^\mathrm{RHF} - \rho^\mathrm{TF}) + \sum_{j=1}^K \int_{|x-R_j|<D} \rho^\mathrm{TF}
\le C+Z,
\]
which proves the theorem.\qed


\section*{Acknowledgments}

The author would like to thank  Shu Nakamura for helpful comments. She also thanks  Heinz Siedentop for many fruitful discussions.
This work was supported by Research Fellow of the JSPS KAKENHI Grant Number 18J13709.

\end{document}